\documentclass[
  twocolumn,
  aps,
  prl,
  showpacs,
  superscriptaddress,
  preprintnumbers,
  letterpaper
]{revtex4}

\usepackage{amsthm}
\usepackage{amsmath}
\usepackage{amssymb}
\usepackage{bbold}
\usepackage{graphicx}
\usepackage{subfigure}

\usepackage{tikz}
\usepgflibrary{shapes.geometric}
\usetikzlibrary{arrows}

\newtheorem{theorem}{Theorem}

\theoremstyle{definition}

\newcommand*{\cA}{\mathcal{A}}

\newcommand*{\cE}{\mathcal{E}}

\newcommand*{\cI}{\mathcal{I}}

\newcommand*{\cM}{\mathcal{M}}
\newcommand*{\cN}{\mathcal{N}}

\newcommand{\ket}[1]{|#1\rangle}
\newcommand{\bra}[1]{\langle #1 |}

\newcommand{\braket}[2]{\langle #1 |#2 \rangle}

\newcommand{\proj}[1]{\ket{#1}\bra{#1}}

\newcommand{\beq}{\begin{equation}}
\newcommand{\eeq}{\end{equation}}
\newcommand{\best}{\begin{equation*}}
\newcommand{\eest}{\end{equation*}}

\DeclareMathOperator{\Tr}{Tr}

\newcommand{\rax}{\rho_{a|x}}
\newcommand{\sax}{\sigma_{a|x}}
\newcommand{\us}{{\textup{\tiny US}}}
\newcommand{\spann}{\operatorname{span}}

\begin{document}

\title{Einstein-Podolsky-Rosen steering provides the advantage  
in entanglement-assisted subchannel discrimination with one-way measurements}
\author{Marco Piani}
\affiliation{Institute for Quantum Computing \& Department of Physics
  and Astronomy, University of Waterloo, 200 University Avenue West, 
  Waterloo, Ontario N2L 3G1, Canada}

\author{John Watrous}
\affiliation{Institute for Quantum Computing \& School of Computer
  Science,
  University of Waterloo, 200 University Avenue West, 
  Waterloo, Ontario N2L 3G1, Canada}


\begin{abstract}
  Steering is the entanglement-based quantum effect that embodies the 
  ``spooky action at a distance'' disliked by Einstein and scrutinized by
  Einstein, Podolsky, and Rosen. 
  Here we provide a necessary and sufficient characterization of steering,
  based on a quantum information processing task: the discrimination of
  branches in a quantum evolution, which we dub
  \emph{subchannel discrimination}.
  We prove that, for any bipartite steerable state, there are instances of the
  quantum subchannel discrimination problem for which this state allows a
  correct discrimination with strictly higher probability than in absence of
  entanglement, even when measurements are restricted to local measurements
  aided by one-way communication. 
  On the other hand, unsteerable states are useless in such conditions, even
  when entangled.
  We also prove that the above steering advantage can be exactly quantified in
  terms of the \emph{steering robustness}, which is a natural measure of the 
  steerability exhibited by the state.
\end{abstract}

\maketitle

Entanglement is a property of distributed quantum systems that does not have a
classical counterpart~\cite{revent}.
On one hand, entanglement challenges our classical, everyday-life intuition
about the physical world; on the other hand, it is the key element in many
quantum information processing tasks~\cite{nielsen2010quantum}.
The strongest feature that entangled systems can exhibit is
non-locality~\cite{revnonloc}.
A weaker feature related to entanglement is \emph{steering}: roughly speaking,
it corresponds to the fact that one party can induce very different ensembles
for the local state of the other party, beyond what is possible based only on a
conceivable classical knowledge about the other party's ``hidden
state''~\cite{wisemanPRL2007,jonesPRA2007}.
Steering embodies the ``spooky action at a distance''---in the words of
Einstein~\cite{einsteinborn}---identified by Schroedinger~\cite{schroedinger},
scrutinized by Einstein, Podolsky, and Rosen~\cite{einstein1935can}, and
formally put on sound ground in~\cite{wisemanPRL2007,jonesPRA2007}.

Not all entangled states are steerable, and not all steerable states exhibit
nonlocality~\cite{wisemanPRL2007,jonesPRA2007}, but states that exhibit
steering allow for the verification of their entanglement in a
\emph{semi-}device independent way: there is no need to trust the devices used
by the steering party, and the ability to determine
the conditional states of the steered party is
sufficient~\cite{wisemanPRL2007,jonesPRA2007,entverification}.
In general, besides its foundational interest, steering is interesting in
practice in bipartite tasks, like quantum key distribution (QKD)~\cite{QKD},
where it is convenient and/or appropriate to trust the devices of one of two
parties, but not necessarily of the other party.
For example, by exploiting steering it is possible to obtain key rates
unachievable in a full device-independent approach~\cite{acin2007_DIQKD}, but
still assuming less about the devices than in a standard QKD
approach~\cite{branciard2012}.
For these reasons, steering has attracted a lot of interest in recent times,
both theoretically and experimentally
~\cite{ecavalcanti2009,sunders2010,ou1992,bowen2003,smith2012conclusive,
ajbennet2012,handchen2012observation,Steinlechner2013,wittmann2012loophole,
bowles2014one,pusey2013,quantifyingsteering,jevtic2013quantum,milne2014quantum,moroder2014steering,schneeloch2013einstein,schneeloch2013violation,schneeloch2014improving},
mostly directed to the verification of steering.
Nonetheless, an answer to the question ``What is steering useful for?''\
that applies to states that exhibit steering can arguably be considered
limited~\cite{branciard2012,entverification}.
Furthermore, the quantification
of steering has just started to be addressed~\cite{quantifyingsteering}.

In this Letter we fully characterize and quantify steering in an operational
way that nicely matches the asymmetric features of steering, and that breaks
new ground in the investigation of the usefulness of steering.
We prove that every steerable state is a resource in a quantum information task
that we dub \emph{subchannel discrimination}, in a practically relevant
scenario where measurements can only be performed locally. 

Subchannel discrimination is the identification of which branch of a quantum
evolution a quantum system undergoes (see Fig.~\ref{fig:subchannel}). 
It is well known that entanglement between a probe and an ancilla can help in
discriminating different
channels~\cite{Kitaev97,Paulsen02,ChildsPR00,D'ArianoPP01,Acin01,
  GiovannettiLM04,GilchristLN05,RosgenW05,Sacchi05,Sacchi05b,Lloyd08}.
In~\cite{PW09} it was proven that actually \emph{every} entangled state is
useful in some instance of the subchannel discrimination
problem. 
Ref.~\cite{MPW10} raised and analyzed the question of whether such an advantage
is preserved when joint measurements on the output probe and the ancilla are
not possible. 
Here we prove that, when only local measurements coordinated by forward
classical communication are possible, \emph{every} steerable state remains
useful, while non-steerable entangled states become useless. 
We further prove that this usefulness, optimized over all instances of the
subchannel discrimination problem, is exactly equal to the 
\emph{robustness of steering}---a natural way of quantifying steering using
techniques similar to the ones used in~\cite{quantifyingsteering}, but based on
the notion of robustness~\cite{robustness,generalizedrobustnesssteiner,
  generalizedrobustnessharrow,geller2014quantifying}. 
We argue that the resulting quantification of steering, besides having
operational interpretations both in terms of resilience to noise and
usefulness, is quantitatively more detailed.
  
\emph{Preliminaries: entanglement and steering.---} 
In the following we will denote by a $\,\hat{}$ (hat) mathematical entities
that are ``normalized.''
So, for example, a positive semidefinite operator with unit trace is a
(normalized) state $\hat{\rho}$.
An \emph{ensemble} $\cE=\{\rho_a\}_a$ for a state $\hat{\rho}$ is a 
collection of \emph{sub}states $\rho_a\leq \hat{\rho}$ such that
$\sum_{a}\rho_a=\hat{\rho}$.
Each substate $\rho_a$ can be seen as being proportional to a normalized state
$\hat{\rho}_a$, $\rho_a=p_a\hat{\rho}_a$, with $p_a=\Tr(\rho_a)$ being the
probability of $\hat{\rho}_a$ in the ensemble.
An \emph{assemblage} $\cA=\{\cE_x\}_x=\{\rax\}_{a,x}$ is a collection
of ensembles $\cE_x$ for the same state $\hat{\rho}$, one for each $x$, i.e.,
$\sum_a\rax=\hat{\rho}$, for all $x$.
For example, $\cE=\{\frac{1}{2}\proj{0},\frac{1}{2}\proj{1}\}$ and
$\cE'=\{\frac{1}{2}\proj{+},\frac{1}{2}\proj{-}\}$, with
$\ket{\pm}:=(\ket{0}\pm\ket{1})/\sqrt{2}$, are both ensembles for the maximally
mixed state $\openone/2$ of a qubit, and taken together they form an assemblage
$\cA=\{\cE,\cE'\}$ for $\openone/2$.

Along similar lines, a \emph{measurement assemblage}
$\cM\cA=\{M_{a|x}\}_{a,x}$
is a collection of positive operators $M_{a|x}\geq 0$ satisfying 
$\sum_{a} M_{a|x} = \openone$ for each $x$.
Such a collection represents one \emph{positive-operator-valued measure} (or
POVM), describing a general quantum measurement, for each $x$.
For a fixed bipartite state $\hat{\rho}_{AB}$, every measurement assemblage on
Alice gives rise to an assemblage on Bob, via
\beq
\label{eq:measurement}
\rho^B_{a|x}=\Tr_A\bigl(M_{a|x}^A\hat{\rho}_{AB}\bigr).
\eeq
On the other hand, every assemblage on Bob $\{\sax\}_{a,x}$ has a quantum
realization \eqref{eq:measurement} for some $\hat{\rho}_{AB}$ satisfying
$\hat{\rho}_{B}=\Tr_A(\hat{\rho}_{AB})=\sum_x \sigma_{a|x}=:\hat{\sigma}_B$ and
for some measurement assemblage $\{M_{a|x}\}_{a,x}$~\cite{Hughston1993}.

An assemblage $\cA=\{\rax\}_{a,x}$ is \emph{unsteerable} if
\begin{equation}
  \label{eq:unsteerable}
  \rax^\us=\sum_\lambda p(\lambda)p(a|x,\lambda)\hat{\sigma}(\lambda)
  =\sum_\lambda p(a|x,\lambda)\sigma(\lambda),
\end{equation}
for all $a, x$, for some probability distribution $p(\lambda)$, conditional
probability distributions $p(a|x,\lambda)$, and states
$\hat{\sigma}(\lambda)$. 
Here $\lambda$ indicates a (hidden) classical random variable, and we
introduced also subnormalized states
$\sigma(\lambda)=p(\lambda)\hat{\sigma}(\lambda)$. 
We observe that every conditional probability distribution $p(a|x,\lambda)$ 
can be written as a convex combination of deterministic conditional probability
distributions: 
$p(a|x,\lambda)=\sum_\nu p(\nu|\lambda)D(a|x,\nu)$, where 
$D(a|x,\nu)=\delta_{a,f_\nu(x)}$ is a deterministic response function labeled 
by $\nu$.
This means that, by a suitable relabeling,
\beq
\label{eq:unsteerabledeterministic}
\rax^\us=\sum_\lambda D(a|x,\lambda)\sigma(\lambda)\quad\forall a, x,
\eeq
where the summation is over labels of deterministic response functions. 
We say that an assemblage $\{\rax\}_{a,x}$ is \emph{steerable} if it is not 
unsteerable.

A \emph{separable} (or \emph{unentangled}) state is one that admits a
decomposition~\cite{werner1989}
\beq
\label{eq:separable}
\hat{\sigma}_{AB}^{\textrm{sep}}
=\sum_\lambda p(\lambda)\hat{\sigma}_A(\lambda)\otimes\hat{\sigma}_B(\lambda),
\eeq
for $\hat{\sigma}_A(\lambda)$, $\hat{\sigma}_B(\lambda)$ local states, 
$\lambda$ a classical label, and $p(\lambda)$ a probability distribution.
A state is \emph{entangled} if it is not separable.
An unsteerable assemblage can always be obtained via \eqref{eq:measurement}
from the separable state $\rho_{AB}=\sum_\lambda
p(\lambda)\proj{\lambda}_A\otimes \hat{\sigma}(\lambda)_B$,
$M_{a|x}=\sum_{\mu}p(a|x,\mu)\proj{\mu}$, and
$\braket{\mu}{\lambda}=\delta_{\mu\lambda}$. 
Most importantly, any separable state can only give raise to unsteerable
assemblages. 
Indeed, for a separable state of the form \eqref{eq:separable}, one has
\[
\sax^\us=\Tr_A(M_{a|x}\sigma_{AB}^{\textrm{sep}})
=\sum_\lambda p(\lambda)p(a|x,\lambda)\sigma_B(\lambda),
\]
with $p(a|x,\lambda)=\Tr_A(M_{a|x}\sigma_A(\lambda))$.
It follows that entanglement is a necessary condition for steerability, and, in
turn, a steerable assemblage is a clear signature of entanglement.
Interestingly, not all entangled states lead to steerable assemblages by the
action of appropriate local measurement
assemblages~\cite{wisemanPRL2007,jonesPRA2007}; we call \emph{steerable states}
those that do, and \emph{unsteerable states} those that do not.
There exist entangled states that are steerable in one direction, e.g., with
Alice measuring, but not in the other, when it is instead Bob trying to steer
(see, e.g.,~\cite{bowles2014one}). 
In the following we will always think of Alice as the steering party, and, for
the sake of brevity we will call (un)steerable the states that can (not) be
steered by Alice.

\emph{Channel and subchannel identification.---} 
A \emph{subchannel} $\Lambda$ is a linear completely positive map that is trace
non-increasing: $\Tr(\Lambda[\rho])\leq\Tr(\rho)$, for all 
states $\rho$.
If  a subchannel $\Lambda$ is trace-preserving, $\Tr(\Lambda[\rho])=\Tr(\rho)$,
for all $\rho$, we use the $\,\hat{}$ notation and say that $\hat{\Lambda}$ is
a \emph{channel}.
An \emph{instrument} $\cI=\{\Lambda_a\}_a$ for a channel $\hat{\Lambda}$ is
a collection of subchannels $\Lambda_a$ such that
$\hat{\Lambda}=\sum_a\Lambda_a$ (see Figure~\ref{fig:subchannel}). 
Any instrument has a physical realization, where the (classical) index
$a$ can be considered available to some
party~\cite{davies1970operational,nielsen2010quantum,michalhorodeckientmeasures}.

\begin{figure}
  \begin{tikzpicture}[>=latex]

    \node[minimum width=6mm, minimum height=6mm] (Rho) at (-2,-0.75) {$\rho$};
    \node (Lambda1) at (2,1) {$\Lambda_1[\rho]$};
    \node (Lambda2) at (2,0) {$\Lambda_2[\rho]$};
    \node (Lambdaa) at (2,-1.75) {$\Lambda_a[\rho]$};

    \draw[->] (Rho.30) -- node[above=1mm, pos=0.65] {$\Lambda_1$} 
    (Lambda1.west);
    \draw[->] (Rho.5) -- node[above=0.5mm, pos=0.65] {$\Lambda_2$} 
    (Lambda2.west);
    \draw[->] (Rho.340) -- node[above=0.5mm, pos=0.65] {$\Lambda_a$} 
    (Lambdaa.west);

    \node at (0.4,-0.5) {$\vdots$};
    \node at (0.4,-1.75) {$\vdots$};
    \node at (2,-0.75) {$\vdots$};
    \node at (2,-2.25) {$\vdots$};

    \node at (1.9,-0.65) [densely dotted, draw=black!90, minimum height=41mm, 
      minimum width=13mm] {};

    \node at (0.4,-0.6) [densely dotted, draw=black!90, minimum height=35mm, 
      minimum width=11mm] {};

    \node at (3.1,-2.4) {$\hat{\Lambda}[\rho]$};
    \node at (-0.5,-2) {$\hat{\Lambda}$};
  \end{tikzpicture}
  \caption{A decomposition of a channel into subchannels can be seen as a
    decomposition of a quantum evolution into \emph{branches} of the
    evolution. If $\{\Lambda_a\}_a$ is an instrument for $\hat{\Lambda}$,
    then we can imagine that the evolution $\rho\mapsto\hat{\Lambda}[\rho]$ has
    branches $\rho\mapsto\Lambda_a[\rho]$, where each branch takes place with
    probability $\Tr(\Lambda_a[\rho])$. The transformation described by the total
    channel $\hat{\Lambda}$ can be seen as the situation where the ``which-branch''
    information is lost. An example of subchannel discrimination problem is that
    of distinguishing between the two quantum evolutions $\Lambda_i[\rho]=K_i
    \rho K_i^\dagger$, $i=0,1$, with $K_0=\proj{0}+\sqrt{1-\gamma}\proj{1}$ and
    $K_1=\sqrt{\gamma}\ket{0}\bra{1}$, corresponding to the so-called amplitude
    damping channel $\hat{\Lambda}=\Lambda_0+\Lambda_1$~\cite{nielsen2010quantum}.}
  \label{fig:subchannel}
\end{figure}
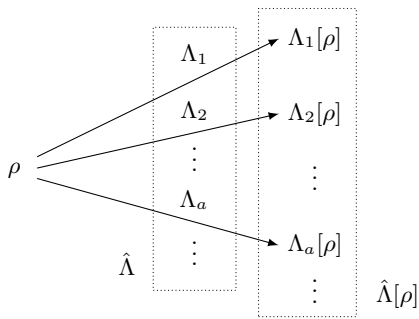

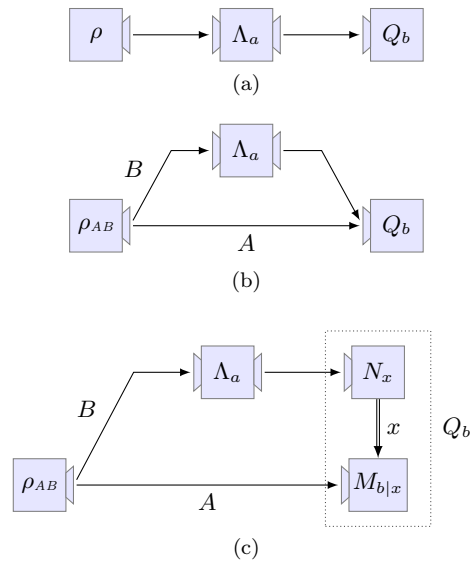
\begin{figure}[!t]
\subfigure[]{
    \begin{tikzpicture}[
        bluebox/.style={draw=black!50, minimum width=7mm, minimum height=7mm,
          fill=blue!10},
        rightnozzle/.style={shape=trapezium, shape border rotate=90, 
          trapezium angle=75, minimum width=5mm, minimum height=1mm, 
          inner sep=0.5mm, trapezium stretches, draw=black!50, fill=blue!10},
        leftnozzle/.style={shape=trapezium, shape border rotate=270, 
          trapezium angle=75, minimum width=5mm, minimum height=1mm, 
          inner sep=0.5mm, trapezium stretches, draw=black!50, fill=blue!10},
        >=latex]
      
      \node (Lambda) at (0,0) [bluebox] {$\Lambda_a$};
      \node (LambdaOut) at ([xshift=-0.4pt]Lambda.east) 
            [rightnozzle, anchor=west] {};

      \node (LambdaIn) at ([xshift=0.4pt]Lambda.west)[leftnozzle,
        anchor=east] {};
      
      \node (Rho) at (-2,0) [bluebox] {$\rho$};
      \node (RhoOut) at ([xshift=-0.4pt]Rho.east) [rightnozzle, anchor=west] {};
      
      \node (Q) at (2,0) [bluebox] {$Q_b$};
      \node (QIn) at ([xshift=0.4pt]Q.west) [leftnozzle, anchor=east] {};
      \draw[->] ([xshift=1pt]RhoOut.east) -- ([xshift=-1pt]LambdaIn.west);
      \draw[->] ([xshift=1pt]LambdaOut.east) -- ([xshift=-1pt]QIn.west);
    \end{tikzpicture}
  \label{fig:NE}
}
\vspace{3mm}
\subfigure[]{
\begin{tikzpicture}[
       bluebox/.style={draw=black!50, minimum width=7mm, minimum height=7mm,
        fill=blue!10},
      rightnozzle/.style={shape=trapezium, shape border rotate=90, 
        trapezium angle=75, minimum width=5mm, minimum height=1mm, 
        inner sep=0.5mm, trapezium stretches, draw=black!50, fill=blue!10},
      leftnozzle/.style={shape=trapezium, shape border rotate=270, 
        trapezium angle=75, minimum width=5mm, minimum height=1mm, 
        inner sep=0.5mm, trapezium stretches, draw=black!50, fill=blue!10},
      >=latex]

    \node (Lambda) at (0,1) [bluebox] {$\Lambda_a$};
    \node (LambdaOut) at ([xshift=-0.4pt]Lambda.east) 
          [rightnozzle, anchor=west] {};
    \node (LambdaIn) at ([xshift=0.4pt]Lambda.west) 
          [leftnozzle, anchor=east] {};

    \node (Rho) at (-2,0) [bluebox] {$\rho_{\textit{\tiny AB}}$};
    \node (RhoOut) at ([xshift=-0.4pt]Rho.east) 
          [rightnozzle, anchor=west] {};

    \node (Q) at (2,0) [bluebox] {$Q_b$};
    \node (QIn) at ([xshift=0.4pt]Q.west) 
          [leftnozzle, anchor=east] {};

    \draw[->] ([xshift=1pt,yshift=2pt]RhoOut.east) -- node[above left] {$B$}
    (-1,1) -- ([xshift=-1pt]LambdaIn.west);
    \draw[->] ([xshift=1pt]LambdaOut.east) -- (1,1) -- 
    ([xshift=-1pt,yshift=2pt]QIn.west);

    \draw[->] ([xshift=1pt]RhoOut.east) -- node[below] {$A$}
    ([xshift=-1pt]QIn.west);

  \end{tikzpicture}
\label{fig:diamond}
}
\subfigure[]{
  \begin{tikzpicture}[
       bluebox/.style={draw=black!50, minimum width=7mm, minimum height=7mm,
        fill=blue!10},
      rightnozzle/.style={shape=trapezium, shape border rotate=90, 
        trapezium angle=75, minimum width=5mm, minimum height=1mm, 
        inner sep=0.5mm, trapezium stretches, draw=black!50, fill=blue!10},
      leftnozzle/.style={shape=trapezium, shape border rotate=270, 
        trapezium angle=75, minimum width=5mm, minimum height=1mm, 
        inner sep=0.5mm, trapezium stretches, draw=black!50, fill=blue!10},
      >=latex]

    \node (Lambda) at (0.5,1.5) [bluebox] {$\Lambda_a$};
    \node (LambdaOut) at ([xshift=-0.4pt]Lambda.east) 
          [rightnozzle, anchor=west] {};
    \node (LambdaIn) at ([xshift=0.4pt]Lambda.west) 
          [leftnozzle, anchor=east] {};

    \node (Rho) at (-2,0) [bluebox] {$\rho_{\textit{\tiny AB}}$};
    \node (RhoOut) at ([xshift=-0.4pt]Rho.east) 
          [rightnozzle, anchor=west] {};

    \node (N) at (2.5,1.5) [bluebox] {$N_x$};
    \node (NIn) at ([xshift=0.4pt]N.west) 
          [leftnozzle, anchor=east] {};

    \node (M) at (2.5,0) [bluebox, inner sep=0.5mm] {$M_{b|x}$};
    \node (MIn) at ([xshift=0.4pt]M.west) 
          [leftnozzle, anchor=east] {};

    \draw[->] ([xshift=1pt,yshift=2pt]RhoOut.east) -- node[above left] {$B$}
    (-0.75,1.5) -- ([xshift=-1pt]LambdaIn.west);
    \draw[->] ([xshift=1pt]LambdaOut.east) -- ([xshift=-1pt]NIn.west);

    \draw[->] ([xshift=1pt]RhoOut.east) -- node[below] {$A$}
    ([xshift=-1pt]MIn.west);

    \draw[->,double] (N.south) -- node[right] {$x$}
    (M.north);

    \node at (2.5,0.75) [densely dotted, draw=black!70, minimum height=26mm, 
      minimum width=14mm] {};

    \node at (3.55,0.75) {$Q_b$};

  \end{tikzpicture}
  \label{fig:oneway}
}
\caption{Different strategies for subchannel discrimination. (a) No
  entanglement is used: a probe, initially in the state $\rho$, undergoes the
  quantum evolution $\hat{\Lambda}$, with branches $\Lambda_a$, and is later
  measured, with an outcome $b$ for the measurement described by the POVM
  $\{Q_b\}_b$, which is the guess for which branch of the evolution actually
  took place. (b) The probe $B$ is potentially entangled with an ancilla $A$;
  the output probe and the ancilla are jointly measured. (c) The probe is still
  potentially entangled with an ancilla, but the final measurement $\{Q_b\}_b$
  is restricted to local measurements on the output probe and the ancilla,
  coordinated  by one-way classical communication
 (single lines represent quantum systems, double lines classical information):
  the outcome $x$ of the measurement $\{N_x\}_x$ performed on the output probe
  is used to decide which measurement $\{M_{b|x}\}_b$ to perform on the
  ancilla.}
\label{fig:fig}
\end{figure}


 Fix an instrument $\{\Lambda_a\}_a$ for a channel $\hat{\Lambda}$, and
 consider a measurement $\{Q_b\}_b$ on the output space of $\hat{\Lambda}$. 
 The joint probability of $\Lambda_a$ and $Q_b$ for input $\rho$ is 
 $p(a,b):=\Tr(Q_b \Lambda_a[\rho])=p(b|a)p(a)$, where 
 $p(a)=\Tr(\Lambda_a[\rho])$ is the probability of the subchannel $\Lambda_a$
 for the given input $\rho$ and $p(b|a)=p(a,b)/p(a)$ is the conditional
 probability of the outcome $b$ given that the subchannel $\Lambda_a$ took 
 place (see Figure~\ref{fig:NE}). 
 The probability of correctly identifying which subchannel was realized is
 \begin{multline}
   p_\textrm{corr}(\{\Lambda_a\}_a,\{Q_b\}_b,\rho)\\
   :=\sum_{a,b}  p(a,b)\delta_{ab}
 =\sum_a \Tr(Q_a \Lambda_a[\rho]).
 \end{multline}

The archetypal case of subchannel discrimination is that of \emph{channel}
discrimination, where $\Lambda_a = p_a \hat{\Lambda}_a$, with channels
$\hat{\Lambda}_a$ and probabilities $p_a$. 
The problem often considered is 
that of telling apart just two
channels $\hat{\Lambda}_0$ and $\hat{\Lambda}_1$, each given with probability
$p_0=p_1=1/2$.
In this case the total (average) channel is simply 
$\hat{\Lambda}=\frac{1}{2}\hat{\Lambda}_0 + \frac{1}{2} \hat{\Lambda}_1$.
The best success probability in identifying subchannels $\{\Lambda_a\}_a$ with
an input $\rho$ is defined as
$p_\textrm{corr}(\{\Lambda_a\}_a,\rho):=\max_{\{Q_b\}_b}p_\textrm{corr}(\{\Lambda_a\}_a,\{Q_b\}_b,\rho)$.
Optimizing also over the input state, one arrives at
$p^{\textup{NE}}_\textrm{corr}(\{\Lambda_a\}_a)
:=\max_{\rho}p_\textrm{corr}(\{\Lambda_a\}_a,\rho)$,
where the superscript $\textup{NE}$ stands for ``no entanglement'' (see
Fig.~\ref{fig:NE}).

Indeed, one may try to improve the success probability by using an entangled
input state $\rho_{AB}$ of an input probe $B$ and an ancilla $A$. 
The guess about which subchannel took place is based on a joint measurement 
of the output probe and the ancilla (see Fig.~\ref{fig:diamond}), with success
probability $p_\textrm{corr}(\{\Lambda^B_a\}_a,\{Q^{AB}_b\}_b,\rho_{AB})$. 
In the latter expression we have explicitly indicated that the subchannels act
non-trivially only on $B$, while input state and measurement pertain to
$AB$.
One can define the optimal probability of success for a scheme that uses
input entanglement and global measurements:
$p^{\textup{E}}_\textrm{corr}(\{\Lambda_a\}_a):=\max_{\rho_{AB}}\max_{\{Q^{AB}_b\}_b}p_\textrm{corr}(\{\Lambda^B_a\}_a,\{Q^{AB}_b\}_b,\rho_{AB})$.
We say that entanglement is useful in discriminating subchannels
$\{\Lambda_a\}_a$ if
$p^{\textup{E}}_\textrm{corr}(\{\Lambda_a\}_a)
>p^{\textup{NE}}_\textrm{corr}(\{\Lambda_a\}_a)$.
It is known that there are instances of subchannel discrimination, already in
the simple setting
$\{\Lambda_a\}_a=\{\frac{1}{2}\hat{\Lambda}_0,\frac{1}{2}\hat{\Lambda}_0\}$,
where $p^{\textup{E}}_\textrm{corr}\approx 1\gg
p^{\textup{NE}}_\textrm{corr}\approx 0$~(see \cite{MPW10} and references
therein).

In~\cite{PW09} it was proven that, for any entangled state $\rho_{AB}$, there
exists a choice $\{\frac{1}{2}\hat{\Lambda}_0,\frac{1}{2}\hat{\Lambda}_1\}$
such that
\[
p_\textrm{corr}\Bigl(\Bigl\{
\frac{1}{2}\hat{\Lambda}_0,\frac{1}{2}\hat{\Lambda}_1\Bigr\},
\rho_{AB}\Bigr)
> p^{\textup{NE}}_\textrm{corr}\Bigl(\Bigl\{\frac{1}{2}\hat{\Lambda}_0,
\frac{1}{2}\hat{\Lambda}_1\Bigr\}\Bigr),
\]
i.e., that every entangled state is useful for the task of (sub)channel
discrimination.
In this sense, every entangled state, independently of how weakly entangled it
is, is a resource. 
Nonetheless, exploiting such a resource may  require arbitrary joint
measurements on the output probe and ancilla~\cite{MPW10}.
From a conceptual perspective, one may want to limit measurements to those that
can be performed by local operations and classical communication
(LOCC), as this makes the input entangled state the only non-local resource. This limitation can be justified also from a practical perspective: LOCC measurements are arguably easier to implement, and might be the only feasible kind of measurements, especially in a scenario where only weakly entangled states can be produced.
We do not know whether every entangled state  stays useful for subchannel
discrimination when measurements are restricted to be LOCC. In the following,
though, we prove that, if the measurements are limited to local operations and
forward communication (one-way LOCC), then only steerable states remain
useful.

\emph{Steerability and subchannel identification by means of restricted 
  measurements.---}  
A Bob-to-Alice one-way LOCC measurement of the form 
$\cM^{B\rightarrow A}=\{Q_a^{B\rightarrow A}\}_a$ has the structure
$Q_a^{B\rightarrow A}=\sum_x M^A_{a|x}\otimes N^B_x$, 
where $\{N^B_x\}_x$ is a measurement on $B$ and $\{M^A_{a|x}\}_{a,x}$ is a
measurement assemblage on $A$. 
We define 
$p^{B\rightarrow A}_\textrm{corr}(\cI,\rho_{AB})
:=\max_{\cM^{B\rightarrow A}}p_\textrm{corr}(\cI^B,\cM^{B\rightarrow A},
\rho_{AB})$ 
as the optimal probability of success in the discrimination of the instrument
$\cI^B=\{\Lambda^B_a\}_a$ by means of the input state $\rho_{AB}$ and one-way
LOCC measurements from $B$ to $A$ (see Fig.~\ref{fig:oneway}). We say that
$\rho_{AB}$ is useful in this restricted-measurement scenario if
$p^{B\rightarrow
  A}_\textrm{corr}(\cI,\rho_{AB})>p^{\textup{NE}}_\textrm{corr}(\cI)$ for some
instrument $\cI$
\footnote{Notice that no bipartite state $\rho_{AB}$ is useful in one-way 
  subchannel identification when the communication goes from the ancilla to the
  output probe. 
  This is because the initial measurement of the ancilla simply creates an
  ensemble of input substates for the channel, and we might as well choose the
  best input to begin with. 
  So, the only one-way communication that may have a non-trivial effect is that
  from the output probe to the ancilla.}. 
We find in general
\begin{multline}
\label{eq:ingeneral}
p_\textrm{corr}(\cI^B,\cM^{B\rightarrow A},\rho_{AB})
=\sum_a\Tr(Q_a^{B\rightarrow A} \Lambda^B_a[\rho_{AB}])\\
=\sum_{a,x} \Tr_B(\Lambda^{\dagger B}_a[N^B_x]  \rax),
\end{multline}
where we used \eqref{eq:measurement}, and $\Lambda^{\dagger}_a$ is the dual of
$\Lambda_a$, defined via $\Tr(X\Lambda_a[Y])=\Tr(\Lambda^\dagger_a[X]Y)$, 
$\forall X, Y$ (assuming $\Lambda_a$ is completely positive). 
If the assemblage $\cA=\{\rax\}_{a,x}$ appearing on the last line of
\eqref{eq:ingeneral} is unsteerable, then we can achieve an equal or better
performance with an uncorrelated probe in the best input state
$\hat{\sigma}(\lambda)$ among the ones appearing in
Eq.~\eqref{eq:unsteerable}.
Thus, if $\rho_{AB}$ is unsteerable, then it is useless for subchannel
discrimination with one-way measurements. 
This applies also to entangled states that are unsteerable, which are
nonetheless useful in channel discrimination with arbitrary
measurements~\cite{PW09}.

We will now prove that every steerable state \emph{is} useful in subchannel
discrimination with one-way-LOCC measurements. 
To state our result in full detail we need to introduce the \emph{steering
  robustness of $\rho_{AB}$},
\beq
\label{eq:steeringrobustnessrho}
R^{A\rightarrow B}_{\textrm{steer}}(\rho_{AB}) := \sup_{\cM\cA} R(\cA),
\eeq
where the supremum is over all measurement assemblages 
$\cM\cA=\{M_{a|x}\}_{a,x}$ on $A$, $R(\cA)$ is the \emph{steering robustness of
  the assemblage $\cA$},
\begin{multline}
\label{eq:steeringrobustnessassemblage}
R(\cA):=\min\biggl\{t\geq 0
\,\bigg|\,\Big\{\frac{\rho_{a|x}+t\,\tau_{a|x}}{1+t}\Big\}_{a,x}
\textrm{ unsteerable},\\
\{\tau_{a|x}\}\textrm{ an assemblage}\biggr\},
\end{multline}
and $\cA$ is obtained from $\rho_{AB}$ with the measurement assemblage
$\cM\cA$ on $A$ (see Eq.~\eqref{eq:measurement}). 
The steering robustness of $\cA$ is a measure of the minimal ``noise'' needed
to destroy the steerability of the assemblage $\cA$, where such noise is in
terms of the mixing with an arbitrary assemblage $\{\tau_{a|x}\}_{a,x}$.
With the notation set, we have the following theorem.
\begin{theorem}
\label{thm:main}
Every steerable state is useful in one-way subchannel discrimination. 
More precisely, it holds
\beq
\label{eq:everysteer}
\sup_{\cI}\frac{p^{B\rightarrow A}_\textrm{corr}(\cI,\rho_{AB})}{p^{\textup{NE}}_\textrm{corr}(\cI)}=R^{A\rightarrow B}_{\textrm{steer}}(\rho_{AB})+1,
\eeq
where the supremum is over all instruments $\cI$.
\end{theorem}
\begin{proof}
Using the definitions \eqref{eq:steeringrobustnessrho} and
\eqref{eq:steeringrobustnessassemblage} it is immediate to verify (see
Appendix)
\[
p_\textrm{corr}(\cI^B,\cM^{B\rightarrow A},\rho_{AB})\leq 
(1+R^{A\rightarrow B}_{\textrm{steer}}(\rho_{AB}))
p^{\textup{NE}}_\textrm{corr}(\cI),
\]
for any $\cM^{B\rightarrow A}$ and any $\cI$.
We will prove next that the bound can be approximated arbitrarily well. 
We will do so by constructing appropriate instances of the subchannel 
discrimination problem. 
To do this, we will need that the steering robustness $R(\cA)$ of any 
assemblage $\cA=\{\rax\}_{a,x}$ can be calculated via semidefinite programming
(SDP)~\cite{boyd2004convex}. 
In particular, in the Appendix we prove that $R(\cA) + 1$ is equal to the
optimal value of the SDP optimization problem
\begin{subequations}
  \label{eq:SDProbustness}
  \begin{align}
    {\text{maximize}}\quad
    & \sum_{a,x}\Tr(F_{a|x}\rho_{a|x}) \\
    \text{subject to}\quad
    \label{eq:SDPcondition1} &  
    \sum_{a,x} D(a|x,\lambda)F_{a|x}\leq\openone\quad\forall \lambda\\
    & F_{a|x}\geq0 \quad\forall a,x,\label{eq:Fpositive}
  \end{align}
\end{subequations}
where the $\lambda$'s are labels for the deterministic response functions.

Now, let $\cM\cA=\{M_{a|x}\}_{a,x}$ be a measurement assemblage on $A$, and
$\cA$ 
the resulting assemblage on $B$.
 Let $F_{a|x}$ be optimal,
 i.e., such that $ \sum_{a,x}\Tr(F_{a|x}\rho_{a|x})=1+R(\cA)$. Define linear maps $\Lambda_a$ via their duals, as
\begin{align}
\label{eq:quantumtoclassical}\Lambda_a^\dagger &= \Lambda_a^\dagger \circ \Pi_{X} & \forall& a,\\
\label{eq:define}\Lambda_a^\dagger\left[\proj{x}\right]&=\alpha F_{a|x} & \forall& a,x.
\end{align}
Here $\circ$ is composition, and $\Pi_X$ indicates the projector onto an orthonormal basis $\{\ket{x}\}$, $x=1,\dots,|X|$, where $|X|$ is the number of settings in the measurement assemblage $\cM\cA$. The constant $\alpha>0$ will be chosen soon. Because of the conditions \eqref{eq:Fpositive}, \eqref{eq:quantumtoclassical}, and \eqref{eq:define}, the $\Lambda_a^\dagger$'s are completely positive linear maps, hence the $\Lambda_a$'s are too; they 
act according to $\Lambda_a[\rho]=\alpha\sum_x \Tr( F_{a|x}\rho)\proj{x}$, and can be seen as subchannels as long as 
$\sum_a \Lambda^\dagger_a[\openone] =\sum_{a,x} \Lambda^\dagger_a[\proj{x}] = \alpha \sum_{a,x} F_{a|x} \leq \openone$,
a condition that can be  satisfied for $\alpha=  \|\sum_{a,x}
F_{a|x}\|_\infty^{-1}$, with $\|\cdot\|_\infty$ the operator norm. 

We can now introduce $N$ additional subchannels, defined as 
$\Lambda_a[\rho] = \frac{1}{N}\Tr((\openone-\sum_a 
\Lambda^\dagger_a[\openone])\rho)\hat{\sigma}_a$, for $a=|A|+1,\ldots,|A|+N$,
where $|A|$ indicates the original number of outcomes for POVMs in $\cM\cA$,
and $\hat{\sigma}_a$ are arbitrary states in a two-dimensional space orthogonal
to $\spann\{\ket{x}\,|\,x=1,\dots,|X|\}$.
The subchannels $\Lambda_a$, $a=1,\dots, |A|+N$ do define an instrument $\cI$ for the trace-preserving channel $\hat{\Lambda}=\sum_{a=1}^{|A|+N} \Lambda_a$, and one can readily (see Appendix)  incorporate the measurement assemblage $\cM\cA$ into 
a one-way LOCC strategy $\cM^{B\rightarrow A}$ such that
$
\alpha (1+R(\cA))
\leq p_\textrm{corr}(\cI^B,\cM^{B\rightarrow A},\rho_{AB})\leq \alpha(1+R(\cA)) +\frac{2}{N}$.
On the other hand, condition \eqref{eq:SDPcondition1} implies (see Appendix)
$\alpha \leq p^{\textup{NE}}_\textrm{corr}(\cI) \leq \alpha + \frac{2}{N}$, so $p_\textrm{corr}(\cI^B,\cM^{B\rightarrow A},\rho_{AB})/p^{\textup{NE}}_\textrm{corr}(\cI)\geq  \frac{1+R(\cA)}{1+2/(\alpha N)}$. The claim follows since 
$N$ can be chosen arbitrarily large.
\end{proof}

\emph{Conclusions.---} We have proven that the steerable states are precisely
those states that are useful for the task of subchannel discrimination with 
feed-forward local measurements. 
This provides a satisfactorily answer to a question left open by~\cite{MPW10}
about the characterization of a large class of entangled states that remain
useful for (sub)channel discrimination with local measurements. 
Most importantly, it provides a full operational characterization---and proof
of usefulness---of steering in terms of a fundamental task, subchannel
discrimination, in a setting---that of restricted measurements---very relevant
from the practical point of view. 
The construction in the proof of Theorem~\ref{thm:main} proves that, for any
measurement assemblage $\cM\cA$ on $A$ such that the corresponding $\cA$
exhibit steering with robustness $R(\cA)>0$, there exist instances of the
subchannel discrimination problem with restricted measurements where the use of
the steerable state ensures a probability of success approximately
$(1+R(\cA))$-fold higher than in the case where no entanglement is used.
Thus, the robustnesses $R(\cA)$ and 
$R^{A\rightarrow B}_{\textrm{steer}}(\rho_{AB})$ have operational meanings not
only in terms of the resilience of steerability versus noise, but in
applicative terms.
Also, they constitute semi-device-independent lower bounds,
\beq
\label{eq:robustnessbounds}
R(\cA)\leq R^{A\rightarrow B}_{\textrm{steer}}(\rho_{AB}) \leq R_g(\rho_{AB}),
\eeq
on the generalized robustness of entanglement
$R_g(\rho_{AB})$~\cite{generalizedrobustnesssteiner,
  generalizedrobustnessharrow},
defined as 
\begin{equation}
\label{eq:entanglementrobustness}
\min\Big\{t\geq 0\,\Big|\,\frac{\rho_{AB}+t\,\tau_{AB}}{1+t}
\textrm{ separable},
\tau\textrm{ a state}\Big\},
\end{equation}
which is an entanglement measure with operational interpretations
itself~\cite{brandaodatta,brandaorobustness}. That \eqref{eq:robustnessbounds}
holds is immediate, given definitions \eqref{eq:steeringrobustnessrho} and
\eqref{eq:steeringrobustnessassemblage} and the fact that a separable state
leads always to unsteerable assemblages. 
Besides these observations, we believe that the way to quantify steerability
that we have introduced is finer-grained than the approach
of~\cite{quantifyingsteering}, while preserving the computational efficiency deriving from the use of semidefinite programming. 
For example, while the so-called \emph{steering weight}
of~\cite{quantifyingsteering} is such that all pure entangled states, however
weekly entangled, are deemed maximally steerable, because of \eqref{eq:robustnessbounds} we know that weakly entangled pure states have small steering robustness~\cite{generalizedrobustnessharrow}. On the other hand, maximally entangled states
$\psi^+_d$ for large local dimension $d$ do have large steering robustness. 
Indeed, in the Appendix we prove that, if $d$ is some power of a prime number, then 
$R^{A\rightarrow B}_{\textrm{steer}}(\psi^{+}_d)\geq \sqrt{d}-2$.

The are many questions that remain open for further investigation: a closed
formula for the steerability robustness of pure (maximally entangled) states;
whether the result of Theorem~\ref{thm:main} can be strengthened to prove that
every steerable state is useful for channel---rather than general
subchannel---discrimination with restricted measurements; whether general LOCC
(rather than one-way LOCC) measurements can restore the usefulness of all
entangled states for (sub)channel discrimination.

\emph{Acknowledgments.---}We acknowledge useful discussions and correspondence
with D. Cavalcanti and M. Pusey. We also acknowledge support from NSERC and
CIFAR.


\begin{thebibliography}{58}
\expandafter\ifx\csname natexlab\endcsname\relax\def\natexlab#1{#1}\fi
\expandafter\ifx\csname bibnamefont\endcsname\relax
  \def\bibnamefont#1{#1}\fi
\expandafter\ifx\csname bibfnamefont\endcsname\relax
  \def\bibfnamefont#1{#1}\fi
\expandafter\ifx\csname citenamefont\endcsname\relax
  \def\citenamefont#1{#1}\fi
\expandafter\ifx\csname url\endcsname\relax
  \def\url#1{\texttt{#1}}\fi
\expandafter\ifx\csname urlprefix\endcsname\relax\def\urlprefix{URL }\fi
\providecommand{\bibinfo}[2]{#2}
\providecommand{\eprint}[2][]{\url{#2}}

\bibitem[{\citenamefont{Horodecki et~al.}(2009)\citenamefont{Horodecki,
  Horodecki, Horodecki, and Horodecki}}]{revent}
\bibinfo{author}{\bibfnamefont{R.}~\bibnamefont{Horodecki}},
  \bibinfo{author}{\bibfnamefont{P.}~\bibnamefont{Horodecki}},
  \bibinfo{author}{\bibfnamefont{M.}~\bibnamefont{Horodecki}},
  \bibnamefont{and}
  \bibinfo{author}{\bibfnamefont{K.}~\bibnamefont{Horodecki}},
  \bibinfo{journal}{Rev. Mod. Phys.} \textbf{\bibinfo{volume}{81}},
  \bibinfo{pages}{865} (\bibinfo{year}{2009}),
  \urlprefix\url{http://link.aps.org/doi/10.1103/RevModPhys.81.865}.

\bibitem[{\citenamefont{Nielsen and Chuang}(2010)}]{nielsen2010quantum}
\bibinfo{author}{\bibfnamefont{M.~A.} \bibnamefont{Nielsen}} \bibnamefont{and}
  \bibinfo{author}{\bibfnamefont{I.~L.} \bibnamefont{Chuang}}
  (\bibinfo{year}{2010}).

\bibitem[{\citenamefont{Brunner et~al.}(2014)\citenamefont{Brunner, Cavalcanti,
  Pironio, Scarani, and Wehner}}]{revnonloc}
\bibinfo{author}{\bibfnamefont{N.}~\bibnamefont{Brunner}},
  \bibinfo{author}{\bibfnamefont{D.}~\bibnamefont{Cavalcanti}},
  \bibinfo{author}{\bibfnamefont{S.}~\bibnamefont{Pironio}},
  \bibinfo{author}{\bibfnamefont{V.}~\bibnamefont{Scarani}}, \bibnamefont{and}
  \bibinfo{author}{\bibfnamefont{S.}~\bibnamefont{Wehner}},
  \bibinfo{journal}{Rev. Mod. Phys.} \textbf{\bibinfo{volume}{86}},
  \bibinfo{pages}{419} (\bibinfo{year}{2014}),
  \urlprefix\url{http://link.aps.org/doi/10.1103/RevModPhys.86.419}.

\bibitem[{\citenamefont{Wiseman et~al.}(2007)\citenamefont{Wiseman, Jones, and
  Doherty}}]{wisemanPRL2007}
\bibinfo{author}{\bibfnamefont{H.~M.} \bibnamefont{Wiseman}},
  \bibinfo{author}{\bibfnamefont{S.~J.} \bibnamefont{Jones}}, \bibnamefont{and}
  \bibinfo{author}{\bibfnamefont{A.~C.} \bibnamefont{Doherty}},
  \bibinfo{journal}{Phys. Rev. Lett.} \textbf{\bibinfo{volume}{98}},
  \bibinfo{pages}{140402} (\bibinfo{year}{2007}),
  \urlprefix\url{http://link.aps.org/doi/10.1103/PhysRevLett.98.140402}.

\bibitem[{\citenamefont{Jones et~al.}(2007)\citenamefont{Jones, Wiseman, and
  Doherty}}]{jonesPRA2007}
\bibinfo{author}{\bibfnamefont{S.~J.} \bibnamefont{Jones}},
  \bibinfo{author}{\bibfnamefont{H.~M.} \bibnamefont{Wiseman}},
  \bibnamefont{and} \bibinfo{author}{\bibfnamefont{A.~C.}
  \bibnamefont{Doherty}}, \bibinfo{journal}{Phys. Rev. A}
  \textbf{\bibinfo{volume}{76}}, \bibinfo{pages}{052116}
  (\bibinfo{year}{2007}),
  \urlprefix\url{http://link.aps.org/doi/10.1103/PhysRevA.76.052116}.

\bibitem[{\citenamefont{Born}(1971)}]{einsteinborn}
\bibinfo{author}{\bibfnamefont{M.}~\bibnamefont{Born}},
  \emph{\bibinfo{title}{The Born-Einstein Letters}}
  (\bibinfo{publisher}{Walker}, \bibinfo{year}{1971}).

\bibitem[{\citenamefont{Schr{\"o}dinger}(1935)}]{schroedinger}
\bibinfo{author}{\bibfnamefont{E.}~\bibnamefont{Schr{\"o}dinger}},
  \bibinfo{journal}{Mathematical Proceedings of the Cambridge Philosophical
  Society} \textbf{\bibinfo{volume}{31}}, \bibinfo{pages}{555}
  (\bibinfo{year}{1935}), ISSN \bibinfo{issn}{1469-8064},
  \urlprefix\url{http://journals.cambridge.org/article_S0305004100013554}.

\bibitem[{\citenamefont{Einstein et~al.}(1935)\citenamefont{Einstein, Podolsky,
  and Rosen}}]{einstein1935can}
\bibinfo{author}{\bibfnamefont{A.}~\bibnamefont{Einstein}},
  \bibinfo{author}{\bibfnamefont{B.}~\bibnamefont{Podolsky}}, \bibnamefont{and}
  \bibinfo{author}{\bibfnamefont{N.}~\bibnamefont{Rosen}},
  \bibinfo{journal}{Physical review} \textbf{\bibinfo{volume}{47}},
  \bibinfo{pages}{777} (\bibinfo{year}{1935}).

\bibitem[{\citenamefont{Cavalcanti et~al.}(2013)\citenamefont{Cavalcanti, Hall,
  and Wiseman}}]{entverification}
\bibinfo{author}{\bibfnamefont{E.~G.} \bibnamefont{Cavalcanti}},
  \bibinfo{author}{\bibfnamefont{M.~J.~W.} \bibnamefont{Hall}},
  \bibnamefont{and} \bibinfo{author}{\bibfnamefont{H.~M.}
  \bibnamefont{Wiseman}}, \bibinfo{journal}{Phys. Rev. A}
  \textbf{\bibinfo{volume}{87}}, \bibinfo{pages}{032306}
  (\bibinfo{year}{2013}),
  \urlprefix\url{http://link.aps.org/doi/10.1103/PhysRevA.87.032306}.

\bibitem[{\citenamefont{Gisin et~al.}(2002)\citenamefont{Gisin, Ribordy,
  Tittel, and Zbinden}}]{QKD}
\bibinfo{author}{\bibfnamefont{N.}~\bibnamefont{Gisin}},
  \bibinfo{author}{\bibfnamefont{G.}~\bibnamefont{Ribordy}},
  \bibinfo{author}{\bibfnamefont{W.}~\bibnamefont{Tittel}}, \bibnamefont{and}
  \bibinfo{author}{\bibfnamefont{H.}~\bibnamefont{Zbinden}},
  \bibinfo{journal}{Rev. Mod. Phys.} \textbf{\bibinfo{volume}{74}},
  \bibinfo{pages}{145} (\bibinfo{year}{2002}),
  \urlprefix\url{http://link.aps.org/doi/10.1103/RevModPhys.74.145}.

\bibitem[{\citenamefont{Ac{\'i}n et~al.}(2007)\citenamefont{Ac{\'i}n, Brunner,
  Gisin, Massar, Pironio, and Scarani}}]{acin2007_DIQKD}
\bibinfo{author}{\bibfnamefont{A.}~\bibnamefont{Ac{\'i}n}},
  \bibinfo{author}{\bibfnamefont{N.}~\bibnamefont{Brunner}},
  \bibinfo{author}{\bibfnamefont{N.}~\bibnamefont{Gisin}},
  \bibinfo{author}{\bibfnamefont{S.}~\bibnamefont{Massar}},
  \bibinfo{author}{\bibfnamefont{S.}~\bibnamefont{Pironio}}, \bibnamefont{and}
  \bibinfo{author}{\bibfnamefont{V.}~\bibnamefont{Scarani}},
  \bibinfo{journal}{Phys. Rev. Lett.} \textbf{\bibinfo{volume}{98}},
  \bibinfo{pages}{230501} (\bibinfo{year}{2007}),
  \urlprefix\url{http://link.aps.org/doi/10.1103/PhysRevLett.98.230501}.

\bibitem[{\citenamefont{Branciard et~al.}(2012)\citenamefont{Branciard,
  Cavalcanti, Walborn, Scarani, and Wiseman}}]{branciard2012}
\bibinfo{author}{\bibfnamefont{C.}~\bibnamefont{Branciard}},
  \bibinfo{author}{\bibfnamefont{E.~G.} \bibnamefont{Cavalcanti}},
  \bibinfo{author}{\bibfnamefont{S.~P.} \bibnamefont{Walborn}},
  \bibinfo{author}{\bibfnamefont{V.}~\bibnamefont{Scarani}}, \bibnamefont{and}
  \bibinfo{author}{\bibfnamefont{H.~M.} \bibnamefont{Wiseman}},
  \bibinfo{journal}{Phys. Rev. A} \textbf{\bibinfo{volume}{85}},
  \bibinfo{pages}{010301} (\bibinfo{year}{2012}),
  \urlprefix\url{http://link.aps.org/doi/10.1103/PhysRevA.85.010301}.

\bibitem[{\citenamefont{Cavalcanti et~al.}(2009)\citenamefont{Cavalcanti,
  Jones, Wiseman, and Reid}}]{ecavalcanti2009}
\bibinfo{author}{\bibfnamefont{E.~G.} \bibnamefont{Cavalcanti}},
  \bibinfo{author}{\bibfnamefont{S.~J.} \bibnamefont{Jones}},
  \bibinfo{author}{\bibfnamefont{H.~M.} \bibnamefont{Wiseman}},
  \bibnamefont{and} \bibinfo{author}{\bibfnamefont{M.~D.} \bibnamefont{Reid}},
  \bibinfo{journal}{Phys. Rev. A} \textbf{\bibinfo{volume}{80}},
  \bibinfo{pages}{032112} (\bibinfo{year}{2009}),
  \urlprefix\url{http://link.aps.org/doi/10.1103/PhysRevA.80.032112}.

\bibitem[{\citenamefont{Saunders et~al.}(2010)\citenamefont{Saunders, Jones,
  Wiseman, and Pryde}}]{sunders2010}
\bibinfo{author}{\bibfnamefont{D.~J.} \bibnamefont{Saunders}},
  \bibinfo{author}{\bibfnamefont{S.~J.} \bibnamefont{Jones}},
  \bibinfo{author}{\bibfnamefont{H.~M.} \bibnamefont{Wiseman}},
  \bibnamefont{and} \bibinfo{author}{\bibfnamefont{G.~J.} \bibnamefont{Pryde}},
  \bibinfo{journal}{Nat Phys} \textbf{\bibinfo{volume}{6}},
  \bibinfo{pages}{845} (\bibinfo{year}{2010}),
  \urlprefix\url{http://dx.doi.org/10.1038/nphys1766}.

\bibitem[{\citenamefont{Ou et~al.}(1992)\citenamefont{Ou, Pereira, Kimble, and
  Peng}}]{ou1992}
\bibinfo{author}{\bibfnamefont{Z.~Y.} \bibnamefont{Ou}},
  \bibinfo{author}{\bibfnamefont{S.~F.} \bibnamefont{Pereira}},
  \bibinfo{author}{\bibfnamefont{H.~J.} \bibnamefont{Kimble}},
  \bibnamefont{and} \bibinfo{author}{\bibfnamefont{K.~C.} \bibnamefont{Peng}},
  \bibinfo{journal}{Phys. Rev. Lett.} \textbf{\bibinfo{volume}{68}},
  \bibinfo{pages}{3663} (\bibinfo{year}{1992}),
  \urlprefix\url{http://link.aps.org/doi/10.1103/PhysRevLett.68.3663}.

\bibitem[{\citenamefont{Bowen et~al.}(2003)\citenamefont{Bowen, Schnabel, Lam,
  and Ralph}}]{bowen2003}
\bibinfo{author}{\bibfnamefont{W.~P.} \bibnamefont{Bowen}},
  \bibinfo{author}{\bibfnamefont{R.}~\bibnamefont{Schnabel}},
  \bibinfo{author}{\bibfnamefont{P.~K.} \bibnamefont{Lam}}, \bibnamefont{and}
  \bibinfo{author}{\bibfnamefont{T.~C.} \bibnamefont{Ralph}},
  \bibinfo{journal}{Phys. Rev. Lett.} \textbf{\bibinfo{volume}{90}},
  \bibinfo{pages}{043601} (\bibinfo{year}{2003}),
  \urlprefix\url{http://link.aps.org/doi/10.1103/PhysRevLett.90.043601}.

\bibitem[{\citenamefont{Smith et~al.}(2012)\citenamefont{Smith, Gillett,
  de~Almeida, Branciard, Fedrizzi, Weinhold, Lita, Calkins, Gerrits, Wiseman
  et~al.}}]{smith2012conclusive}
\bibinfo{author}{\bibfnamefont{D.~H.} \bibnamefont{Smith}},
  \bibinfo{author}{\bibfnamefont{G.}~\bibnamefont{Gillett}},
  \bibinfo{author}{\bibfnamefont{M.~P.} \bibnamefont{de~Almeida}},
  \bibinfo{author}{\bibfnamefont{C.}~\bibnamefont{Branciard}},
  \bibinfo{author}{\bibfnamefont{A.}~\bibnamefont{Fedrizzi}},
  \bibinfo{author}{\bibfnamefont{T.~J.} \bibnamefont{Weinhold}},
  \bibinfo{author}{\bibfnamefont{A.}~\bibnamefont{Lita}},
  \bibinfo{author}{\bibfnamefont{B.}~\bibnamefont{Calkins}},
  \bibinfo{author}{\bibfnamefont{T.}~\bibnamefont{Gerrits}},
  \bibinfo{author}{\bibfnamefont{H.~M.} \bibnamefont{Wiseman}},
  \bibnamefont{et~al.}, \bibinfo{journal}{Nature communications}
  \textbf{\bibinfo{volume}{3}}, \bibinfo{pages}{625} (\bibinfo{year}{2012}).

\bibitem[{\citenamefont{Bennet et~al.}(2012)\citenamefont{Bennet, Evans,
  Saunders, Branciard, Cavalcanti, Wiseman, and Pryde}}]{ajbennet2012}
\bibinfo{author}{\bibfnamefont{A.~J.} \bibnamefont{Bennet}},
  \bibinfo{author}{\bibfnamefont{D.~A.} \bibnamefont{Evans}},
  \bibinfo{author}{\bibfnamefont{D.~J.} \bibnamefont{Saunders}},
  \bibinfo{author}{\bibfnamefont{C.}~\bibnamefont{Branciard}},
  \bibinfo{author}{\bibfnamefont{E.~G.} \bibnamefont{Cavalcanti}},
  \bibinfo{author}{\bibfnamefont{H.~M.} \bibnamefont{Wiseman}},
  \bibnamefont{and} \bibinfo{author}{\bibfnamefont{G.~J.} \bibnamefont{Pryde}},
  \bibinfo{journal}{Phys. Rev. X} \textbf{\bibinfo{volume}{2}},
  \bibinfo{pages}{031003} (\bibinfo{year}{2012}),
  \urlprefix\url{http://link.aps.org/doi/10.1103/PhysRevX.2.031003}.

\bibitem[{\citenamefont{H{\"a}ndchen et~al.}(2012)\citenamefont{H{\"a}ndchen,
  Eberle, Steinlechner, Samblowski, Franz, Werner, and
  Schnabel}}]{handchen2012observation}
\bibinfo{author}{\bibfnamefont{V.}~\bibnamefont{H{\"a}ndchen}},
  \bibinfo{author}{\bibfnamefont{T.}~\bibnamefont{Eberle}},
  \bibinfo{author}{\bibfnamefont{S.}~\bibnamefont{Steinlechner}},
  \bibinfo{author}{\bibfnamefont{A.}~\bibnamefont{Samblowski}},
  \bibinfo{author}{\bibfnamefont{T.}~\bibnamefont{Franz}},
  \bibinfo{author}{\bibfnamefont{R.~F.} \bibnamefont{Werner}},
  \bibnamefont{and} \bibinfo{author}{\bibfnamefont{R.}~\bibnamefont{Schnabel}},
  \bibinfo{journal}{Nature Photonics} \textbf{\bibinfo{volume}{6}},
  \bibinfo{pages}{596} (\bibinfo{year}{2012}).

\bibitem[{\citenamefont{Steinlechner et~al.}(2013)\citenamefont{Steinlechner,
  Bauchrowitz, Eberle, and Schnabel}}]{Steinlechner2013}
\bibinfo{author}{\bibfnamefont{S.}~\bibnamefont{Steinlechner}},
  \bibinfo{author}{\bibfnamefont{J.}~\bibnamefont{Bauchrowitz}},
  \bibinfo{author}{\bibfnamefont{T.}~\bibnamefont{Eberle}}, \bibnamefont{and}
  \bibinfo{author}{\bibfnamefont{R.}~\bibnamefont{Schnabel}},
  \bibinfo{journal}{Phys. Rev. A} \textbf{\bibinfo{volume}{87}},
  \bibinfo{pages}{022104} (\bibinfo{year}{2013}),
  \urlprefix\url{http://link.aps.org/doi/10.1103/PhysRevA.87.022104}.

\bibitem[{\citenamefont{Wittmann et~al.}(2012)\citenamefont{Wittmann, Ramelow,
  Steinlechner, Langford, Brunner, Wiseman, Ursin, and
  Zeilinger}}]{wittmann2012loophole}
\bibinfo{author}{\bibfnamefont{B.}~\bibnamefont{Wittmann}},
  \bibinfo{author}{\bibfnamefont{S.}~\bibnamefont{Ramelow}},
  \bibinfo{author}{\bibfnamefont{F.}~\bibnamefont{Steinlechner}},
  \bibinfo{author}{\bibfnamefont{N.~K.} \bibnamefont{Langford}},
  \bibinfo{author}{\bibfnamefont{N.}~\bibnamefont{Brunner}},
  \bibinfo{author}{\bibfnamefont{H.~M.} \bibnamefont{Wiseman}},
  \bibinfo{author}{\bibfnamefont{R.}~\bibnamefont{Ursin}}, \bibnamefont{and}
  \bibinfo{author}{\bibfnamefont{A.}~\bibnamefont{Zeilinger}},
  \bibinfo{journal}{New Journal of Physics} \textbf{\bibinfo{volume}{14}},
  \bibinfo{pages}{053030} (\bibinfo{year}{2012}).

\bibitem[{\citenamefont{Bowles et~al.}(2014)\citenamefont{Bowles, V{\'e}rtesi,
  Quintino, and Brunner}}]{bowles2014one}
\bibinfo{author}{\bibfnamefont{J.}~\bibnamefont{Bowles}},
  \bibinfo{author}{\bibfnamefont{T.}~\bibnamefont{V{\'e}rtesi}},
  \bibinfo{author}{\bibfnamefont{M.~T.} \bibnamefont{Quintino}},
  \bibnamefont{and} \bibinfo{author}{\bibfnamefont{N.}~\bibnamefont{Brunner}},
  \bibinfo{journal}{Physical Review Letters} \textbf{\bibinfo{volume}{112}},
  \bibinfo{pages}{200402} (\bibinfo{year}{2014}).

\bibitem[{\citenamefont{Pusey}(2013)}]{pusey2013}
\bibinfo{author}{\bibfnamefont{M.~F.} \bibnamefont{Pusey}},
  \bibinfo{journal}{Phys. Rev. A} \textbf{\bibinfo{volume}{88}},
  \bibinfo{pages}{032313} (\bibinfo{year}{2013}),
  \urlprefix\url{http://link.aps.org/doi/10.1103/PhysRevA.88.032313}.

\bibitem[{\citenamefont{Skrzypczyk et~al.}(2014)\citenamefont{Skrzypczyk,
  Navascu\'es, and Cavalcanti}}]{quantifyingsteering}
\bibinfo{author}{\bibfnamefont{P.}~\bibnamefont{Skrzypczyk}},
  \bibinfo{author}{\bibfnamefont{M.}~\bibnamefont{Navascu\'es}},
  \bibnamefont{and}
  \bibinfo{author}{\bibfnamefont{D.}~\bibnamefont{Cavalcanti}},
  \bibinfo{journal}{Phys. Rev. Lett.} \textbf{\bibinfo{volume}{112}},
  \bibinfo{pages}{180404} (\bibinfo{year}{2014}),
  \urlprefix\url{http://link.aps.org/doi/10.1103/PhysRevLett.112.180404}.

\bibitem[{\citenamefont{Jevtic et~al.}(2013)\citenamefont{Jevtic, Pusey,
  Jennings, and Rudolph}}]{jevtic2013quantum}
\bibinfo{author}{\bibfnamefont{S.}~\bibnamefont{Jevtic}},
  \bibinfo{author}{\bibfnamefont{M.~F.} \bibnamefont{Pusey}},
  \bibinfo{author}{\bibfnamefont{D.}~\bibnamefont{Jennings}}, \bibnamefont{and}
  \bibinfo{author}{\bibfnamefont{T.}~\bibnamefont{Rudolph}},
  \bibinfo{journal}{arXiv preprint arXiv:1303.4724}  (\bibinfo{year}{2013}).

\bibitem[{\citenamefont{Milne et~al.}(2014)\citenamefont{Milne, Jevtic,
  Jennings, Wiseman, and Rudolph}}]{milne2014quantum}
\bibinfo{author}{\bibfnamefont{A.}~\bibnamefont{Milne}},
  \bibinfo{author}{\bibfnamefont{S.}~\bibnamefont{Jevtic}},
  \bibinfo{author}{\bibfnamefont{D.}~\bibnamefont{Jennings}},
  \bibinfo{author}{\bibfnamefont{H.}~\bibnamefont{Wiseman}}, \bibnamefont{and}
  \bibinfo{author}{\bibfnamefont{T.}~\bibnamefont{Rudolph}},
  \bibinfo{journal}{arXiv preprint arXiv:1403.0418}  (\bibinfo{year}{2014}).

\bibitem[{\citenamefont{Moroder et~al.}(2014)\citenamefont{Moroder, Gittsovich,
  Huber, and G{\"u}hne}}]{moroder2014steering}
\bibinfo{author}{\bibfnamefont{T.}~\bibnamefont{Moroder}},
  \bibinfo{author}{\bibfnamefont{O.}~\bibnamefont{Gittsovich}},
  \bibinfo{author}{\bibfnamefont{M.}~\bibnamefont{Huber}}, \bibnamefont{and}
  \bibinfo{author}{\bibfnamefont{O.}~\bibnamefont{G{\"u}hne}},
  \bibinfo{journal}{arXiv preprint arXiv:1405.0262}  (\bibinfo{year}{2014}).

\bibitem[{\citenamefont{Schneeloch
  et~al.}(2013{\natexlab{a}})\citenamefont{Schneeloch, Broadbent, Walborn,
  Cavalcanti, and Howell}}]{schneeloch2013einstein}
\bibinfo{author}{\bibfnamefont{J.}~\bibnamefont{Schneeloch}},
  \bibinfo{author}{\bibfnamefont{C.~J.} \bibnamefont{Broadbent}},
  \bibinfo{author}{\bibfnamefont{S.~P.} \bibnamefont{Walborn}},
  \bibinfo{author}{\bibfnamefont{E.~G.} \bibnamefont{Cavalcanti}},
  \bibnamefont{and} \bibinfo{author}{\bibfnamefont{J.~C.}
  \bibnamefont{Howell}}, \bibinfo{journal}{Physical Review A}
  \textbf{\bibinfo{volume}{87}}, \bibinfo{pages}{062103}
  (\bibinfo{year}{2013}{\natexlab{a}}).

\bibitem[{\citenamefont{Schneeloch
  et~al.}(2013{\natexlab{b}})\citenamefont{Schneeloch, Dixon, Howland,
  Broadbent, and Howell}}]{schneeloch2013violation}
\bibinfo{author}{\bibfnamefont{J.}~\bibnamefont{Schneeloch}},
  \bibinfo{author}{\bibfnamefont{P.~B.} \bibnamefont{Dixon}},
  \bibinfo{author}{\bibfnamefont{G.~A.} \bibnamefont{Howland}},
  \bibinfo{author}{\bibfnamefont{C.~J.} \bibnamefont{Broadbent}},
  \bibnamefont{and} \bibinfo{author}{\bibfnamefont{J.~C.}
  \bibnamefont{Howell}}, \bibinfo{journal}{Physical review letters}
  \textbf{\bibinfo{volume}{110}}, \bibinfo{pages}{130407}
  (\bibinfo{year}{2013}{\natexlab{b}}).

\bibitem[{\citenamefont{Schneeloch et~al.}(2014)\citenamefont{Schneeloch,
  Broadbent, and Howell}}]{schneeloch2014improving}
\bibinfo{author}{\bibfnamefont{J.}~\bibnamefont{Schneeloch}},
  \bibinfo{author}{\bibfnamefont{C.~J.} \bibnamefont{Broadbent}},
  \bibnamefont{and} \bibinfo{author}{\bibfnamefont{J.~C.}
  \bibnamefont{Howell}}, \bibinfo{journal}{Physics Letters A}
  (\bibinfo{year}{2014}).

\bibitem[{\citenamefont{Kitaev}(1997)}]{Kitaev97}
\bibinfo{author}{\bibfnamefont{A.}~\bibnamefont{Kitaev}},
  \bibinfo{journal}{Russ. Math. Surv.} \textbf{\bibinfo{volume}{52}},
  \bibinfo{pages}{1191} (\bibinfo{year}{1997}).

\bibitem[{\citenamefont{Paulsen}(2002)}]{Paulsen02}
\bibinfo{author}{\bibfnamefont{V.}~\bibnamefont{Paulsen}},
  \emph{\bibinfo{title}{Completely Bounded Maps and Operator Algebras}},
  Cambridge Studies in Advanced Mathematics (\bibinfo{publisher}{Cambridge
  University Press}, \bibinfo{year}{2002}).

\bibitem[{\citenamefont{Childs et~al.}(2000)\citenamefont{Childs, Preskill, and
  Renes}}]{ChildsPR00}
\bibinfo{author}{\bibfnamefont{A.}~\bibnamefont{Childs}},
  \bibinfo{author}{\bibfnamefont{J.}~\bibnamefont{Preskill}}, \bibnamefont{and}
  \bibinfo{author}{\bibfnamefont{J.}~\bibnamefont{Renes}}, \bibinfo{journal}{J.
  Mod. Opt.} \textbf{\bibinfo{volume}{47}}, \bibinfo{pages}{155}
  (\bibinfo{year}{2000}).

\bibitem[{\citenamefont{D'Ariano et~al.}(2001)\citenamefont{D'Ariano, LoPresti,
  and Paris}}]{D'ArianoPP01}
\bibinfo{author}{\bibfnamefont{G.~M.} \bibnamefont{D'Ariano}},
  \bibinfo{author}{\bibfnamefont{P.}~\bibnamefont{LoPresti}}, \bibnamefont{and}
  \bibinfo{author}{\bibfnamefont{M.~G.~A.} \bibnamefont{Paris}},
  \bibinfo{journal}{Phys. Rev. Lett.} \textbf{\bibinfo{volume}{87}},
  \bibinfo{pages}{270404} (\bibinfo{year}{2001}).

\bibitem[{\citenamefont{Acin}(2001)}]{Acin01}
\bibinfo{author}{\bibfnamefont{A.}~\bibnamefont{Acin}}, \bibinfo{journal}{Phys.
  Rev. Lett.} \textbf{\bibinfo{volume}{87}}, \bibinfo{pages}{177901}
  (\bibinfo{year}{2001}).

\bibitem[{\citenamefont{Giovannetti et~al.}(2004)\citenamefont{Giovannetti,
  Lloyd, and Maccone}}]{GiovannettiLM04}
\bibinfo{author}{\bibfnamefont{V.}~\bibnamefont{Giovannetti}},
  \bibinfo{author}{\bibfnamefont{S.}~\bibnamefont{Lloyd}}, \bibnamefont{and}
  \bibinfo{author}{\bibfnamefont{L.}~\bibnamefont{Maccone}},
  \bibinfo{journal}{Science} \textbf{\bibinfo{volume}{306}},
  \bibinfo{pages}{1330} (\bibinfo{year}{2004}).

\bibitem[{\citenamefont{Gilchrist et~al.}(2005)\citenamefont{Gilchrist,
  Langford, and Nielsen}}]{GilchristLN05}
\bibinfo{author}{\bibfnamefont{A.}~\bibnamefont{Gilchrist}},
  \bibinfo{author}{\bibfnamefont{N.~K.} \bibnamefont{Langford}},
  \bibnamefont{and} \bibinfo{author}{\bibfnamefont{M.~A.}
  \bibnamefont{Nielsen}}, \bibinfo{journal}{Phys. Rev. A}
  \textbf{\bibinfo{volume}{71}}, \bibinfo{pages}{062310}
  (\bibinfo{year}{2005}).

\bibitem[{\citenamefont{Rosgen and Watrous}(2005)}]{RosgenW05}
\bibinfo{author}{\bibfnamefont{B.}~\bibnamefont{Rosgen}} \bibnamefont{and}
  \bibinfo{author}{\bibfnamefont{J.}~\bibnamefont{Watrous}}, in
  \emph{\bibinfo{booktitle}{Proc. 20th Ann. Conf. Comp. Compl.}}
  (\bibinfo{year}{2005}), pp. \bibinfo{pages}{344--354}.

\bibitem[{\citenamefont{Sacchi}(2005{\natexlab{a}})}]{Sacchi05}
\bibinfo{author}{\bibfnamefont{M.~F.} \bibnamefont{Sacchi}},
  \bibinfo{journal}{Phys. Rev. A} \textbf{\bibinfo{volume}{71}},
  \bibinfo{pages}{062340} (\bibinfo{year}{2005}{\natexlab{a}}).

\bibitem[{\citenamefont{Sacchi}(2005{\natexlab{b}})}]{Sacchi05b}
\bibinfo{author}{\bibfnamefont{M.~F.} \bibnamefont{Sacchi}},
  \bibinfo{journal}{Phys. Rev. A} \textbf{\bibinfo{volume}{72}},
  \bibinfo{pages}{014305} (\bibinfo{year}{2005}{\natexlab{b}}).

\bibitem[{\citenamefont{Lloyd}(2008)}]{Lloyd08}
\bibinfo{author}{\bibfnamefont{S.}~\bibnamefont{Lloyd}},
  \bibinfo{journal}{Science} \textbf{\bibinfo{volume}{321}},
  \bibinfo{pages}{1463} (\bibinfo{year}{2008}).

\bibitem[{\citenamefont{Piani and Watrous}(2009)}]{PW09}
\bibinfo{author}{\bibfnamefont{M.}~\bibnamefont{Piani}} \bibnamefont{and}
  \bibinfo{author}{\bibfnamefont{J.}~\bibnamefont{Watrous}},
  \bibinfo{journal}{Phys. Rev. Lett.} \textbf{\bibinfo{volume}{102}},
  \bibinfo{eid}{250501} (pages~\bibinfo{numpages}{4}) (\bibinfo{year}{2009}).

\bibitem[{\citenamefont{Matthews et~al.}(2010)\citenamefont{Matthews, Piani,
  and Watrous}}]{MPW10}
\bibinfo{author}{\bibfnamefont{W.}~\bibnamefont{Matthews}},
  \bibinfo{author}{\bibfnamefont{M.}~\bibnamefont{Piani}}, \bibnamefont{and}
  \bibinfo{author}{\bibfnamefont{J.}~\bibnamefont{Watrous}},
  \bibinfo{journal}{Phys. Rev. A} \textbf{\bibinfo{volume}{82}},
  \bibinfo{pages}{032302} (\bibinfo{year}{2010}),
  \urlprefix\url{http://link.aps.org/doi/10.1103/PhysRevA.82.032302}.

\bibitem[{\citenamefont{Vidal and Tarrach}(1999)}]{robustness}
\bibinfo{author}{\bibfnamefont{G.}~\bibnamefont{Vidal}} \bibnamefont{and}
  \bibinfo{author}{\bibfnamefont{R.}~\bibnamefont{Tarrach}},
  \bibinfo{journal}{Phys.Rev. A} \textbf{\bibinfo{volume}{59}},
  \bibinfo{pages}{141} (\bibinfo{year}{1999}).

\bibitem[{\citenamefont{Steiner}(2003)}]{generalizedrobustnesssteiner}
\bibinfo{author}{\bibfnamefont{M.}~\bibnamefont{Steiner}},
  \bibinfo{journal}{Phys. Rev. A} \textbf{\bibinfo{volume}{67}},
  \bibinfo{pages}{054305} (\bibinfo{year}{2003}).

\bibitem[{\citenamefont{Harrow and
  Nielsen}(2003)}]{generalizedrobustnessharrow}
\bibinfo{author}{\bibfnamefont{A.}~\bibnamefont{Harrow}} \bibnamefont{and}
  \bibinfo{author}{\bibfnamefont{M.}~\bibnamefont{Nielsen}},
  \bibinfo{journal}{Phys. Rev. A} \textbf{\bibinfo{volume}{68}},
  \bibinfo{pages}{012308} (\bibinfo{year}{2003}).

\bibitem[{\citenamefont{Geller and Piani}(2014)}]{geller2014quantifying}
\bibinfo{author}{\bibfnamefont{J.}~\bibnamefont{Geller}} \bibnamefont{and}
  \bibinfo{author}{\bibfnamefont{M.}~\bibnamefont{Piani}},
  \bibinfo{journal}{arXiv preprint arXiv:1401.8197}  (\bibinfo{year}{2014}).

\bibitem[{\citenamefont{Hughston et~al.}(1993)\citenamefont{Hughston, Jozsa,
  and Wootters}}]{Hughston1993}
\bibinfo{author}{\bibfnamefont{L.~P.} \bibnamefont{Hughston}},
  \bibinfo{author}{\bibfnamefont{R.}~\bibnamefont{Jozsa}}, \bibnamefont{and}
  \bibinfo{author}{\bibfnamefont{W.~K.} \bibnamefont{Wootters}},
  \bibinfo{journal}{Physics Letters A} \textbf{\bibinfo{volume}{183}},
  \bibinfo{pages}{14 } (\bibinfo{year}{1993}), ISSN \bibinfo{issn}{0375-9601},
  \urlprefix\url{http://www.sciencedirect.com/science/article/pii/037596019390%
8809}.

\bibitem[{\citenamefont{Werner}(1989)}]{werner1989}
\bibinfo{author}{\bibfnamefont{R.~F.} \bibnamefont{Werner}},
  \bibinfo{journal}{Phys. Rev. A} \textbf{\bibinfo{volume}{40}},
  \bibinfo{pages}{4277} (\bibinfo{year}{1989}),
  \urlprefix\url{http://link.aps.org/doi/10.1103/PhysRevA.40.4277}.

\bibitem[{\citenamefont{Davies and Lewis}(1970)}]{davies1970operational}
\bibinfo{author}{\bibfnamefont{E.~B.} \bibnamefont{Davies}} \bibnamefont{and}
  \bibinfo{author}{\bibfnamefont{J.~T.} \bibnamefont{Lewis}},
  \bibinfo{journal}{Communications in Mathematical Physics}
  \textbf{\bibinfo{volume}{17}}, \bibinfo{pages}{239} (\bibinfo{year}{1970}).

\bibitem[{\citenamefont{Horodecki}(2001)}]{michalhorodeckientmeasures}
\bibinfo{author}{\bibfnamefont{M.}~\bibnamefont{Horodecki}},
  \bibinfo{journal}{Quantum information and computation}
  \textbf{\bibinfo{volume}{1}}, \bibinfo{pages}{3} (\bibinfo{year}{2001}).

\bibitem[{\citenamefont{Boyd and Vandenberghe}(2004)}]{boyd2004convex}
\bibinfo{author}{\bibfnamefont{S.~P.} \bibnamefont{Boyd}} \bibnamefont{and}
  \bibinfo{author}{\bibfnamefont{L.}~\bibnamefont{Vandenberghe}},
  \emph{\bibinfo{title}{Convex optimization}} (\bibinfo{publisher}{Cambridge
  university press}, \bibinfo{year}{2004}).

\bibitem[{\citenamefont{Brand{\~a}o and Datta}(2011)}]{brandaodatta}
\bibinfo{author}{\bibfnamefont{F.~G. S.~L.} \bibnamefont{Brand{\~a}o}}
  \bibnamefont{and} \bibinfo{author}{\bibfnamefont{N.}~\bibnamefont{Datta}},
  \bibinfo{journal}{IEEE Transactions on Information Theory}
  \textbf{\bibinfo{volume}{57}}, \bibinfo{pages}{1754} (\bibinfo{year}{2011}).

\bibitem[{\citenamefont{Brand{\~a}o}(2007)}]{brandaorobustness}
\bibinfo{author}{\bibfnamefont{F.~G. S.~L.} \bibnamefont{Brand{\~a}o}},
  \bibinfo{journal}{Phys. Rev. A} \textbf{\bibinfo{volume}{76}},
  \bibinfo{pages}{030301(R)} (\bibinfo{year}{2007}).

\bibitem[{\citenamefont{Vidal and Werner}(2002)}]{negativity}
\bibinfo{author}{\bibfnamefont{G.}~\bibnamefont{Vidal}} \bibnamefont{and}
  \bibinfo{author}{\bibfnamefont{R.~F.} \bibnamefont{Werner}},
  \bibinfo{journal}{Phys. Rev. A} \textbf{\bibinfo{volume}{65}},
  \bibinfo{pages}{032314} (\bibinfo{year}{2002}),
  \urlprefix\url{http://link.aps.org/doi/10.1103/PhysRevA.65.032314}.

\bibitem[{\citenamefont{Jiang et~al.}(2013)\citenamefont{Jiang, Piani, and
  Caves}}]{jiang2013ancilla}
\bibinfo{author}{\bibfnamefont{Z.}~\bibnamefont{Jiang}},
  \bibinfo{author}{\bibfnamefont{M.}~\bibnamefont{Piani}}, \bibnamefont{and}
  \bibinfo{author}{\bibfnamefont{C.~M.} \bibnamefont{Caves}},
  \bibinfo{journal}{Quantum information processing}
  \textbf{\bibinfo{volume}{12}}, \bibinfo{pages}{1999} (\bibinfo{year}{2013}).

\bibitem[{\citenamefont{Bengtsson}(2006)}]{bengtsson2006three}
\bibinfo{author}{\bibfnamefont{I.}~\bibnamefont{Bengtsson}},
  \bibinfo{journal}{arXiv preprint quant-ph/0610216}  (\bibinfo{year}{2006}).

\bibitem[{\citenamefont{Jozsa and Schlienz}(2000)}]{jozsaSchlienz2000}
\bibinfo{author}{\bibfnamefont{R.}~\bibnamefont{Jozsa}} \bibnamefont{and}
  \bibinfo{author}{\bibfnamefont{J.}~\bibnamefont{Schlienz}},
  \bibinfo{journal}{Phys. Rev. A} \textbf{\bibinfo{volume}{62}},
  \bibinfo{pages}{012301} (\bibinfo{year}{2000}).

\end{thebibliography}

\appendix

\section{Robustness as semidefinite program}

Inspired by the work of Pusey~\cite{pusey2013} and Skrzypczyk et
al.~\cite{quantifyingsteering}, we are going to prove that calculating the
steering robustness  $R(\cA)$ of an assemblage $\cA=\{\rax\}_{a,x}$
falls under the umbrella of semidefinite
programming (SDP)~\cite{boyd2004convex}.

By definition, see Eq.~\eqref{eq:steeringrobustnessassemblage}, $R(\cA)$ is 
the minimum positive $t$ such that
\[
\rax=(1+t)\sax^\us - t\tau_{a|x},\quad \forall a,x,
\]
with $\{\sax^\us\}_{a,x}$ an unsteerable assemblage and $\{\tau_{a|x}\}_{a,x}$
an arbitrary assemblage. 
Notice that, since $\{\rax\}_{a,x}$ and $\{\sax^\us\}_{a,x}$ are assemblages,
$\tau_{a|x}=\big((1+t)\sax^\us - \rax\big)/t$ is automatically an assemblage 
as long as
\beq
\label{eq:positive}
(1+t)\sax^\us \geq \rax, ,\quad \forall a,x,
\eeq
Since $\{\sax^\us\}_{a,x}$ is unsteerable, see 
Eq.~\eqref{eq:unsteerabledeterministic}, we can rewrite Eq.~\eqref{eq:positive}
as the condition
\[
(1+t)\sum_\lambda D(a|x,\lambda)\sigma_\lambda   \geq \rax, \quad \forall a,x,
\]
where the $\sigma_\lambda$'s are subnormalized states, and the sum is over all
the deterministic strategies to output $a$ given $x$.
If we consider that the factor $(1+t)$ can be absorbed into the
$\sigma_\lambda$'s (so that they are generally unnormalized, rather
subnormalized), we realize that $R(\cA)+1$ can be characterized as the solution to
\begin{equation}
  \label{eq:SDProbustnessprimalappendix}
  \begin{aligned}
    & {\text{minimize}}
    & & \sum_{\lambda} \Tr(\sigma_\lambda) \\
    & \text{subject to}
    & &  \sum_\lambda D(a|x,\lambda)\sigma_\lambda\geq\rax\quad\forall a,x\\
    & & & \sigma_\lambda\geq0 \quad\forall \lambda
  \end{aligned}
\end{equation}
This is an example of SDP optimization 
problem~\cite{boyd2004convex}.
For our purposes, the \emph{primal problem} of an SDP is an optimization
problem cast as 
\begin{equation*}
\begin{aligned}
& {\text{minimize}}
& & \langle C,X \rangle\\
& \text{subject to}
& & \Phi[X] \geq B \\
& & & X\geq 0,
\end{aligned}
\end{equation*}
where:
\begin{itemize}
\item $\langle C,X \rangle$ is the objective function;
\item $B$ and $C$ are given Hermitian matrices;
\item $X$ is the matrix variable on which to optimize;
\item $\langle X,Y\rangle :=\Tr(X^\dagger Y)$ is the Hilbert-Schmidt inner product;
\item $\Phi$ is a given Hermiticity-preserving linear map.
\end{itemize}
The \emph{dual problem} provides a lower bound to the objective function of the
primal problem. 
The dual problem is given by 
\begin{equation*}
\begin{aligned}
& {\text{maximize}}
& & \langle B,Y \rangle\\
& \text{subject to}
& & \Phi^\dagger [Y] \leq C \\
& & & Y\geq 0, 
\end{aligned}
\end{equation*}
where $\Phi^\dagger$ is the dual of $\Phi$ with respect to the Hilbert-Schmidt
inner product, and $Y$ is another matrix variable.

One says that \emph{strong duality} holds when the optimal values of the primal
and dual problems coincide. 
Strong duality holds in many cases, and in particular under the Slater
conditions that 
(i) the primal and dual problems are both feasible, and moreover the
primal problem is \emph{strictly feasible}, meaning that there is a positive
definite $X>0$ such that $\Phi[X]>B$, or 
(ii) the primal and dual problems are both feasible, and moreover the
dual problem is strictly feasible, meaning that there is a $Y>0$ such that 
$\Phi^\dagger[Y]<C$. 
In case (i), not only do the primal and dual values coincide, but there must
exist $Y_\textrm{opt}$ that achieves the optimal value for the dual problem;
and similarly, in the case (ii), there must exist $X_\textrm{opt}$ that
achieves the optimal value in the primal problem.

In our case
\begin{gather*}
C=\openone,\quad B=\operatorname{diag}(\rax)_{a,x},\\
\Phi[X] = \operatorname{diag}\left(\sum_\lambda D(a|x,\lambda)X_\lambda\right)_{a,x}
\end{gather*}
where $\operatorname{diag}(\cdot)_{a,x}$ indicates a block-diagonal matrix
whose diagonal blocks are labeled by $a,x$, and the $X_\lambda$'s are the
diagonal blocks of $X$, labeled by $\lambda$.
Thus, we have 
$\Phi^\dagger[Y]=\operatorname{diag}
\left(\sum_{a,x}D(a|x,\lambda)Y_{a|x}\right)_\lambda$, and the dual of the
primal problem \eqref{eq:SDProbustnessprimalappendix} reads
\begin{subequations}
  \label{eq:SDProbustnessappendix}
  \begin{alignat}{3}
    & {\text{maximize}}
    & &\sum_{a,x}\Tr(F_{a|x}\rho_{a|x}) \\
    & \text{subject to }
    \label{eq:SDPcondition1appendix}& &  
    \sum_{a,x} D(a|x,\lambda)F_{a|x}\leq\openone\quad\forall \lambda\\
    & & & F_{a|x}\geq0 \quad\forall a,x,\label{eq:Fpositiveappendix}
\end{alignat}
\end{subequations}

It is easy to verify that both Slater conditions hold in our case.
For instance, one can take $\sigma_\lambda = 2\openone$ for all $\lambda$, 
and $F_{a|x}=\frac{\openone}{|X|+1}$ for all $a,x$, with 
$|X|$ being the number of possible values for $x$. 
Thus, there exist $F_{a|x}=F^\textrm{opt}_{a|x}$ satisfying the constraints 
of Eq.~\eqref{eq:SDProbustnessappendix} and such that 
$\sum_{a,x}\Tr(F_{a|x}\rho_{a|x})=1+R(\cA)$.
  
We remark that the optimal $F_{a|x}$ can always be chosen to saturate
\eqref{eq:SDPcondition1appendix}. 
That is, there is a deterministic strategy $D(a|x,\lambda)$ and a normalized
pure state $\ket{\phi}$ such that
\beq
\label{eq:saturation}
\sum_{a,x} D(a|x,\lambda)\bra{\phi} F_{a|x}\ket{\phi} 
= \bra{\phi}\openone \ket{\phi} =1
\eeq
This is because otherwise it is always possible to increase (in operator sense)
some $F_{a|x}$'s, still maintaining the optimal value for the objective 
function (which is operator monotone in the $F_{a|x}$'s).

\section{Details of the proof of Theorem~\ref{thm:main}}

The claimed upper bound,
\[
p_\textrm{corr}(\cI^B,\cM^{B\rightarrow A},\rho_{AB})\leq 
(1+R^{A\rightarrow B}_{\textrm{steer}}(\rho_{AB}))
p^{\textup{NE}}_\textrm{corr}(\cI),
\]
can be proved using \eqref{eq:ingeneral} and definitions
\eqref{eq:steeringrobustnessrho} and \eqref{eq:steeringrobustnessassemblage}:
\begin{multline*}
  p_\textrm{corr}(\cI^B,\cM^{B\rightarrow A},\rho_{AB})\\
  \begin{aligned}
    &=\sum_{a,x} \Tr_B(\Lambda^{\dagger B}_a[N^B_x]  \rax)\\
    &\leq (1+R(\cA))\sum_{a,x} \Tr_B(\Lambda^{\dagger B}_a[N^B_x]  \sax^\us)\\
    &\quad- R(\cA)\sum_{a,x} \Tr_B(\Lambda^{\dagger B}_a[N^B_x]  
    \tau_{a|x}^\us)\\
    &\leq  (1+R(\cA)) p^{\textup{NE}}_\textrm{corr}(\cI)\\
    &\leq (1+R^{A\rightarrow B}_{\textrm{steer}}(\rho_{AB}))
    p^{\textup{NE}}_\textrm{corr}(\cI).
  \end{aligned}
\end{multline*}

On the other hand, suppose that $\cM\cA=\{M_{a|x}\}_{a,x}$, where
$a=1,\ldots,|A|$ and $x=1,\ldots, |X|$, is a measurement assemblage on $A$ such
that the corresponding assemblage
$\cA=\{\rho_{a|x}=\Tr_A(M_{a|x}^A\rho_{AB})\}_{a,x}$ is steerable. 
Let $F_{a|x}\geq 0$ be the operators optimal for
\eqref{eq:SDProbustnessappendix}, such that
$\sum_{a,x}\Tr(F_{a|x}\rho_{a|x})=1+R(\cA)$. 
In the proof of Theorem~\ref{thm:main} of the main text we defined subchannels
$\Lambda_a$ that act as
\begin{multline}
\label{eq:subchannelconstruction}
\Lambda_a[\rho]=\\
\begin{cases}
\alpha \sum_{x=1}^{|X|} \Tr(\rho F_{a|x}) \proj{x} & 1\leq a \leq |A|\\
\frac{1}{N}\Tr((\openone-\sum_{a=1}^{|A|} 
\Lambda^\dagger_a[\openone])\rho)\hat{\sigma}_a & |A|+1\leq a \leq |A|+N,
\end{cases}
\end{multline}
where $\alpha= \|\sum_{a,x} F_{a|x}\|_\infty^{-1}>0$, and the 
$\hat{\sigma}_a$, $a=|A|+1,\ldots, |A|+N$, are arbitrary (normalized) states 
in a two-dimensional subspace orthogonal to
$\spann\{\ket{x}\,|\,x=1,\dots,|X|\}$. 
It is immediate to check that 
$\Tr\left(\sum_{a=1}^{|A|+N} \Lambda_a[\rho]\right) = \Tr(\rho)$ (by 
construction), so $\cI=\{\Lambda_a\}_{a=1,\ldots,|A|+N}$ is an instrument for 
the channel $\sum_{a=1}^{|A|+N} \Lambda_a$.

Let $\sigma_{AB}$ be an arbitrary bipartite state on $AB$, and let
$\cM\cA^{B\rightarrow A}=\{Q_a\}_a^{B\rightarrow A}$ be an arbitrary one-way
measurement from $B$ to $A$, i.e., $Q_a^{B\rightarrow A}=\sum_y
M'^A_{a|y}\otimes N'^B_y$, to guess which subchannel was actually realized. 
Notice that $y$ in the latter expression potentially varies in an arbitrary
range, different from the range $\{1,\ldots,|X|\}$ for the parameter $x$ of the
fixed measurement assemblage $\cM\cA$. 
Nonetheless we observe that
\[
\Lambda_a = \Pi'_X\circ\Lambda_a
\]
for $a=1,\ldots,|A|+N$, where $\circ$ is composition, and 
\[
\Pi'_X[\tau] =\sum_{x=1}^{|X|} \proj{x}\tau\proj{x}+\Pi^\perp \tau \Pi^\perp,
\]
with $\Pi^\perp$ the projector onto the two-dimensional space orthogonal to
$\spann\{\ket{x}\,|\,x=1,\dots,|X|\}$ that supports the arbitrary qubits states
$\hat{\sigma}_a$, $a=|A|+1,\ldots,|A|+N$. 
Also,
\[
\Lambda^B_a[\sigma_{AB}]= \frac{1}{N}\left(\sigma_A-\sum_{a'=1}^{|A|} 
\Tr_B( \Lambda^B_{a'}[\sigma_{AB}])\right)\otimes \hat{\sigma}^B_a,
\]
for $a=|A|+1,\ldots,|A|+N$.
This implies that, for whatever input $\sigma_{AB}$, the optimal
$Q_a^{B\rightarrow A}$ can be chosen to have the form
\begin{multline}
\label{eq:optimalQa}
Q_a^{B\rightarrow A}\\
=
\begin{cases}
\sum_{x=1}^{|X|} M'^A_{a|x}\otimes \proj{x}^B & 1\leq a \leq |A|\\
\openone^A \otimes N^B_a, & |A|+1\leq a \leq |A|+N,
\end{cases}
\end{multline}
with $ \quad\Pi^\perp N_a \Pi^\perp = N_a$, for  $|A|+1\leq a \leq |A|+N$, a
POVM on the orthogonal qubit space. 
Omitting a detailed and straightforward proof of this, we instead provide the
following intuition: For the subchannels \eqref{eq:subchannelconstruction}, the
best local measurement on the output probe is one that first of all
discriminates between the space  $\spann\{\ket{x}\,|\,x=1,\dots,|X|\}$ and the
orthogonal qubit space. 
If the probe is found in the space $\spann\{\ket{x}\,|\,x=1,\dots,|X|\}$, the
probe is then measured in the basis $\{\ket{x}\,|\,x=1,\dots,|X|\}$ and the
result if forwarded to decide which measurement to perform on the ancilla: this
is optimal because, in this subspace, the output probe is already dephased in
the basis $\{\ket{x}\,|\,x=1,\dots,|X|\}$. 
If the probe is instead found in the orthogonal qubit space, there is no
information to be gained from the ancilla, since, for the state of the probe to
have support in the orthogonal qubit space, the probe must have been discarded
and prepared in one of the random qubits states $\hat{\sigma}_a$. 
So, in this case, the ancilla is necessarily decorrelated and its state
independent of the specific $\Lambda_a$, $a=|A|+1,\ldots,|A|+N$, that has been
realized; thus the optimal guess about said $\Lambda_a$ can be made as soon as
the output probe is measured.

Then, for an optimal $\cM^{B\rightarrow A}=\{Q_a^{B\rightarrow A}\}_a$ of the
form \eqref{eq:optimalQa}, we find in general
\[
\begin{aligned}
  p_\textrm{corr}(\cI^B,\cM^{B\rightarrow A},\sigma_{AB})\hspace{-3.6cm}\\
  &=\sum_{a=1}^{|A|+N} \Tr(Q_a^{B\rightarrow A} \Lambda^B_a[\sigma_{AB}])\\
  &=\sum_{a=1}^{|A|} \Tr(Q_a^{B\rightarrow A} 
  \Lambda^B_a[\sigma_{AB}])+\sum_{a=|A|+1}^{|A|+N}  \Tr(Q_a^{B\rightarrow A} 
  \Lambda^B_a[\sigma_{AB}])\\
  &=\sum_{a=1}^{|A|} \sum_{x=1}^{|X|}\Tr(M'^{A}_{a|x} \otimes 
  \proj{x}_B\Lambda^B_a[\sigma_{AB}])\\
  &\quad +\left(1-\sum_{a=1}^{|A|} 
  \Tr( \Lambda^B_a[\sigma_{AB}])\right)\frac{1}{N}\sum_{a=|A|+1}^{|A|+N}
  \Tr(N_a \hat{\sigma}_{a})\\
  &=\sum_{a=1}^{|A|} \sum_{x=1}^{|X|}\Tr(\Lambda^\dagger_a[\proj{x}]
  \sigma_{a|x})\\
  &\quad +\left(1-\sum_{a=1}^{|A|} \Tr( \Lambda^B_a[\sigma_{AB}])\right)
  \frac{1}{N}\sum_{a=|A|+1}^{|A|+N}\Tr(N_a \hat{\sigma}_{a}),\\
\end{aligned}
\]
with $\sigma_{a|x}=\Tr_A(M'_{a|x}\sigma_{AB})$. By construction it holds
\[
\Lambda^\dagger_a[\proj{x}]=\alpha F_{a|x}
\]
for $1\leq a \leq |A|$ and $1\leq x\leq |X|$ (see Eq.~\eqref{eq:define}),
therefore
\beq
\begin{aligned}
p_\textrm{corr}(\cI^B,\cM^{B\rightarrow A},\rho_{AB})\hspace{-3.5cm}\\
&=\alpha \sum_{a=1}^{|A|} \sum_{x=1}^{|X|}\Tr(F_{a|x}\sigma_{a|x})\\
&\quad +\left(1-\sum_{a=1}^{|A|} \Tr( \Lambda^B_a[\sigma_{AB}])\right)\frac{1}{N}\sum_{a=|A|+1}^{|A|+N}\Tr(N_a \hat{\sigma}_{a})\\
\label{eq:boundadditionalsubchannels}&\leq \alpha \sum_{a=1}^{|A|} \sum_{x=1}^{|X|}\Tr(F_{a|x}\sigma_{a|x}) + \frac{2}{N}.
\end{aligned}
\eeq
In the last line we used
\[
\left(1-\sum_{a=1}^{|A|} \Tr( \Lambda^B_a[\sigma_{AB}])\right)\leq 1
\]
and
\begin{multline}
  \frac{1}{N}\sum_{a=|A|+1}^{|A|+N}\Tr(N_a \hat{\sigma}_{a})
  \leq \frac{1}{N}\Tr\left(\sum_{a=|A|+1}^{|A|+N} N_a\right)\\
  \leq \frac{1}{N} \Tr(\Pi^\perp) =\frac{2}{N}.
\end{multline}
It is clear that if $\sigma_{AB}=\rho_{AB}$ and $M'_{a|x}= M_{a|x}$ in
\eqref{eq:optimalQa}, so that $\sigma_{a|x}=\rho_{a|x}$, then we have
\[
1+R(\cA)\leq p_\textrm{corr}(\cI^B,\cM^{B\rightarrow A},\rho_{AB})\leq 
1+R(\cA) + \frac{2}{N}.
\]
It remains to prove that
\beq
\label{eq:lastclaim}
\alpha \leq p^{\textup{NE}}_\textrm{corr}(\cI) \leq \alpha + \frac{2}{N}.
\eeq
This is readily verified by considering that \eqref{eq:SDPcondition1appendix}
can be saturated, as argued at the end of the previous section
(see~\eqref{eq:saturation}), for an optimal solution of the SDP problem. So we
have that for some deterministic $D(a|x,\lambda)$ and some uncorrelated input
state $\ket{\phi}$ to the channel,
\[
\begin{split}
1&=\sum_{a=1}^{|A|}\sum_{x=1}^{|X|} D(a|x,\lambda)\bra{\phi} F_{a|x}\ket{\phi}\\
&=\frac{1}{\alpha}\sum_{a=1}^{|A|}\sum_{x=1}^{|X|} D(a|x,\lambda)\bra{\phi} 
\Lambda_a^\dagger[\proj{x}]\ket{\phi}\\
&=\frac{1}{\alpha} \sum_{a=1}^{|A|} \Tr\left(\left(\sum_{x: D(a|x,\lambda)=1}
\proj{x}\right) \Lambda_a [\proj{\phi}]\right)\\
&=\frac{1}{\alpha} \sum_{a=1}^{|A|} \Tr\left(M''_a 
\Lambda_a [\proj{\phi}]\right),\\
\end{split}
\]
having defined $M''_a:=\sum_{x: D(a|x,\lambda)=1} \proj{x}$.
Considering also the subchannels $\Lambda_a$, $a=|A|+1,\ldots,|A|+N$, and
bounding their contribution to the probability of success as in
\eqref{eq:boundadditionalsubchannels}, we arrive at \eqref{eq:lastclaim}.

\section{On the scaling of the steerability of maximally entangled states}

We have argued that
\[
R^{A\rightarrow B}_{\textrm{steer}}(\rho_{AB}) \leq R_g(\rho_{AB}).
\]
where $R_g(\rho_{AB})$ is the generalized entanglement robustness
\eqref{eq:entanglementrobustness}.
Indeed, let $\tau_{AB}$ be optimal for the generalized entanglement robustness,
i.e., suppose
\[
\sigma_{AB}=\frac{\rho_{AB}+ R_g(\rho_{AB})\tau_{AB}}{1+R_g(\rho_{AB})}
\]
is separable. 
Then $\sax = \Tr_A(M_{a|x}\sigma_{AB})$ is unsteerable for any measurement
assemblage $\{M_{a|x}\}_{a,x}$, proving that $R_g(\rho_{AB})$ is an upper bound
to $R^{A\rightarrow B}_{\textrm{steer}}(\rho_{AB})$ (see
Eq.~\eqref{eq:steeringrobustnessrho}).
This means that, if a state is weakly entangled with respect to $R_g$, it is also weakly steerable with respect to $R^{A\rightarrow B}_{\textrm{steer}}$. In~\cite{generalizedrobustnessharrow} it was proven that, for any bipartite
pure state
\[
\ket{\psi}_{AB}=\sum_i \sqrt{p_i} \ket{i}_A\ket{i}_B,
\]
here in its Schmidt decomposition, the generalized entanglement robustness is
equal to
\[
R_g(\proj{\psi}_{AB})=\left(\sum_i \sqrt{p_i}\right)^2 - 1=2\cN(\proj{\psi}_{AB}),
\]
where $\cN$ is the negativity of entanglement~\cite{negativity}. In particular, then, for a maximally entangled state in dimension $d\times d$,
$\ket{\psi^+_d}_{AB}=\frac{1}{\sqrt{d}}\sum_{i=1}^d\ket{i}_A\ket{i}_B$, one has
\[
R^{A\rightarrow B}_{\textrm{steer}}(\psi^+_{d,AB}) \leq R_g(\psi^+_{d,AB})=d-1,
\]
having used the shorthand notation $\psi^+_{d,AB}=\proj{\psi^+_d}_{AB}$.

We conclude by providing a lower bound on $R^{A\rightarrow B}_{\textrm{steer}}(\psi^+_{d,AB})$ for $d$ a power of a prime number. 
We will use techniques similar to the ones used in the examples of~\cite{jiang2013ancilla}.

Fix $d$ to be the power of a prime number. Then we know that there there are $d+1$ mutually unbiased bases, i.e., $d+1$ orthonormal sets $\{\ket{\psi_{a|x}}\}_{a=1,\dots,d}$, one for each $x=1,\ldots, d+1$, such that~\cite{bengtsson2006three}
\[
|\braket{\psi_{a|x}}{\psi_{b|y}}|
=
\begin{cases}
\delta_{a,b} & x=y \\
\frac{1}{\sqrt{d}} & x \neq y
\end{cases}
\]
We will consider a measurement assemblage $\{M_{a|x}=\proj{\psi_{a|x}}\}_{a,x}$. Suppose $\rho_{AB}=\psi^+_{d,AB}$. We have
\[
\rho^B_{a|x}=\Tr_A(M^A_{a|x} \psi^+_{d,AB})=\frac{1}{d}\proj{\psi^*_{a|x}}
\]
Here $\ket{\psi^*_{a|x}}$ indicates orthonormal vectors whose coefficients in the local basis $\{\ket{i}_B\}$ are the complex conjugate of the coefficients of $\ket{\psi_{a|x}}$ in the local basis $\{\ket{i}_A\}$. Thus, the bases $\{\ket{\psi^*_{a|x}}\}_{a=1,\dots,d}$ are still mutually unbiased.

We want to lower bound the steering robustness of $\{\rax^B\}_{a,x}$, which in turn will give us a lower bound on $R^{A\rightarrow B}_{\textrm{steer}}(\psi^+_{d,AB})$. To do this, we use a specific choice for the $F_{a|x}$'s in \eqref{eq:SDProbustnessappendix}. We choose $F_{a|x}=\beta \proj{\psi^*_{a|x}}$, where $\beta>0$ will be fixed to satisfy \eqref{eq:SDPcondition1appendix} (condition \eqref{eq:Fpositiveappendix} is  satisfied for any $\beta\geq 0$), i.e.,
\[
\left\|\sum_{a,x} D(a|x,\lambda) F_{a|x}\right\|_\infty \leq 1
\]
for all deterministic $D(a|x,\lambda)$. With our choice of $F_{a|x}$, this can be achieved by taking
\beq
\label{eq:betamin}
\begin{split}
\beta
&\leq\left(\max_\lambda\left\|\sum_x  \proj{\psi^*_{f_\lambda(x)|x}}\right\|_\infty\right)^{-1}
\end{split}
\eeq
where the maximum is over all functions $f_\lambda:\{1,\ldots,d+1\}\rightarrow \{1,\ldots,d\}$, labeled by $\lambda$. 
To estimate the right hand side of \eqref{eq:betamin}, we will use the fact~\cite{jozsaSchlienz2000} that, for
\[
\ket{\gamma}_{CD}=\sum_{x=1}^{d+1}\ket{\psi^*_{f_\lambda(x)|x}}_C\ket{x}_D,
\]
where $\{\ket{x}\}_{x=1.\ldots,d+1}$ is an orthonormal basis, the spectrum of
\[
\Tr_D(\proj{\gamma}_{CD})=\sum_x  \proj{\psi^*_{f_\lambda(x)|x}}.
\]
is the same as the spectrum of
\[
\begin{aligned}
\Tr_C(\proj{\gamma}_{CD})
&=\sum_{x,y}\braket{\psi^*_{f_\lambda(x)|x}}{\psi^*_{f_\lambda(y)|y}}
\ket{y}\bra{x}\\
&=\sum_x \ket{x}\bra{x} + \frac{1}{\sqrt{d}}\sum_{x\neq y} 
e^{i \phi_{x,y}} \ket{y}\bra{x}\\
&=\left(1-\frac{1}{\sqrt{d}}\right)\openone + 
\frac{1}{\sqrt{d}}\sum_{x, y} e^{i \phi_{x,y}} \ket{y}\bra{x}
\end{aligned}
\]
where $\phi_{x,y}$ are real numbers representing phases. 
Thus, we have
\[
\begin{split}
\left\|\sum_x  \proj{\psi^*_{f_\lambda(x)|x}}\right\|_\infty\hspace{-3cm}\\
&=\left\|\sum_{x,y}\braket{\psi^*_{f_\lambda(x)|x}}{\psi^*_{f_\lambda(y)|y}}\ket{y}\bra{x}\right\|_\infty\\
&\leq \left(1-\frac{1}{\sqrt{d}}\right) + \frac{1}{\sqrt{d}}\left\| \sum_{x, y} 
e^{i \phi_{x,y}} \ket{y}\bra{x} \right\|_\infty\\
&\leq \left(1-\frac{1}{\sqrt{d}}\right) + \frac{1}{\sqrt{d}}\left\| \sum_{x, y} 
e^{i \phi_{x,y}} \ket{y}\bra{x} \right\|_2\\
&= \left(1-\frac{1}{\sqrt{d}}\right) + \frac{1}{\sqrt{d}} (d+1)\\
 &=1+\sqrt{d}.
\end{split}
\]
Since this estimate is independent of $\lambda$, we can take $\beta= 1/(\sqrt{d}+1)$. Hence, we conclude that, for $d$ the power of a prime number,
\begin{multline}
R^{A\rightarrow B}_{\textrm{steer}}(\psi^+_{d,AB})\\
\begin{aligned}
&\geq R\left(\left\{\frac{1}{d}\proj{\psi^*_{a|x}}\right\}\right)\\
&\geq \sum_{a,x} \Tr\left(\left(\frac{1}{d}\proj{\psi^*_{a|x}}\right)
\left(\frac{1}{\sqrt{d}+1}\proj{\psi^*_{a|x}}\right)\right) -1 \\
&= \frac{1}{d(\sqrt{d}+1)}(d(d+1)) - 1\\
&= \sqrt{d}\frac{\sqrt{d}-1}{\sqrt{d}+1}\\
&\geq \sqrt{d}-2.
\end{aligned}
\end{multline}

\end{document}